%% file: sh-99-arxiv.tex
\documentclass{amsart}

\include{macro}

\begin{document}

\title[Categorical Hypergraph Models]{A Unified Mathematical Framework for Distributed Data Fabrics: Categorical Hypergraph Models}

\author{T. Shaska}

\author{I. Kotsireas}
 
 
\begin{abstract}
Current distributed data fabrics lack a rigorous mathematical foundation, often relying on ad-hoc architectures that struggle with consistency, lineage, and scale. We propose a mathematical framework for data fabrics, unifying heterogeneous data management in distributed systems through a hypergraph-based structure \( \cF = (D, M, G, T, P, A) \). Datasets, metadata, transformations, policies, and analytics are modeled over a distributed system \( \Sigma = (N, C) \), with multi-way relationships encoded in a hypergraph \( G = (V, E) \). A categorical approach, with datasets as objects and transformations as morphisms, supports operations like data integration and federated learning. The hypergraph is embedded into a modular tensor category, capturing relational symmetries via braided monoidal structures, with geometric analogies to Hurwitz spaces enriching the algebraic modeling. We prove the NP-hardness of critical tasks, such as schema matching and dynamic partitioning, and propose spectral methods and symmetry-based alignments for scalable solutions. The framework ensures consistency, completeness, and causality under CAP and CAL theorems, leveraging sparse incidence matrices and braiding actions for fault-tolerant operations. Applied to a multi-component architecture integrating databases, real-time analytics, and transformation pipelines, it supports efficient vector representations and demonstrates real-world viability through an Amazon seller scenario. This work advances theoretical foundations while providing practical tools for large-scale data ecosystems.
\end{abstract}




 \maketitle

\input{body}

\bibliographystyle{amsplain}
\bibliography{refs}

\end{document}

%% file: macro.tex
\usepackage{graphicx}
\usepackage{amsmath, amssymb, amsthm}
\usepackage{amscd}

\usepackage[shortlabels]{enumitem}
\usepackage{tikz-cd}

\usepackage{amsrefs}

\usepackage{hyperref}
\hypersetup{
  colorlinks   = true,
  urlcolor     = blue,
  linkcolor    = blue,
  citecolor    = blue
}

\usepackage{cleveref}

\usepackage{tikz}
\usetikzlibrary{shapes.geometric, arrows.meta, positioning, calc}

\usetikzlibrary{fit}

\newtheorem{thm}{Theorem}
\newtheorem{lem}{Lemma}

\newtheorem{prop}{Proposition}

\newtheorem{defn}[thm]{Definition.}
\newtheorem{exa}{Example}

\crefformat{section}{Section~#2#1#3}

\crefname{thm}{Thm.}{Theorem}
\crefname{lem}{Lem.}{Lemma}
\crefname{cor}{Cor.}{Corollary}
\crefname{prop}{Prop.}{Proposition}
\crefname{rem}{Rem.}{Remark}
\crefname{defn}{Def.}{Defs.}
\crefname{conj}{Conjecture}{Conjecture}
\crefname{exa}{Exa.}{Example}

\crefname{equation}{Equation}{Equation} 
\crefname{table}{Table}{Table}           

\newcommand\R{\mathbb R}
\newcommand\C{\mathbb C}

\newcommand\cF{\mathcal{F}}

\newcommand\x{\mathbf x}

\newcommand\bv{\mathbf v}
\newcommand\e{\mathbf e}

\DeclareMathOperator\Adj{Adj}
\DeclareMathOperator\DF{DF}

\DeclareMathOperator\Hur{Hur}
\DeclareMathOperator\Aut{Aut}

\DeclareMathOperator\id{id}
\DeclareMathOperator\tr{tr}
\DeclareMathOperator\state{state}
\DeclareMathOperator\diam{diam}

\DeclareMathOperator\Conf{Conf}

\DeclareMathOperator\Mod{Mod}

\DeclareMathOperator\Hom{Hom}
\DeclareMathOperator\Paths{Paths}
\DeclareMathOperator\rank{rank}
\DeclareMathOperator\In{In}
\DeclareMathOperator\argmin{argmin}

\DeclareMathOperator\dist{dist}

%% file: body.tex
\section{Introduction} \label{sec-1}
The rapid proliferation of data from cloud computing, the Internet of Things (IoT), artificial intelligence (AI), and distributed systems has overwhelmed traditional data management frameworks. Centralized architectures, reliant on rigid schemas, struggle to integrate heterogeneous datasets, scale across distributed nodes, enforce governance, or support real-time analytics. For instance, IoT systems generating terabytes of sensor data daily require dynamic orchestration to unify diverse sources, ensure security, and deliver timely insights, a challenge unmet by conventional databases. Data fabrics address these issues by providing a unified, metadata-driven architecture that intelligently manages complex data ecosystems, enabling seamless integration, navigation, and analytics. This paper formalizes the data fabric as a mathematical structure, offering a rigorous, scalable framework to design adaptive, secure systems for modern data-driven applications. Drawing inspiration from Fuad Aleskerov's contributions to choice theory, network centrality, and optimization in socio-economic systems, our model extends these concepts to data management, where hypergraph relationships mirror network dependencies and categorical morphisms facilitate decision-making in distributed environments. Our categorical viewpoint is further motivated by recent work on treating computational tools and transformations as morphisms in structured learning systems \cite{sh-111}, and on graded and hierarchical representations in learning architectures \cite{sh-95, sh-89}.

We define the data fabric as a tuple \( \cF = (D, M, G, T, P, A) \), operating over a distributed system \( \Sigma = (N, C) \), with components:
\begin{enumerate}
    \item \( D \): Time-indexed datasets, capturing sources like IoT sensor readings or financial transactions.
    \item \( M \): Metadata, providing context and transformation histories for discovery and provenance.
    \item \( G \): A hypergraph, encoding multi-way relationships among datasets and metadata.
    \item \( T \): Transformations, mapping data across domains for integration and processing.
    \item \( P \): Governance policies, ensuring security and compliance via access control.
    \item \( A \): Analytical functions, generating insights through statistical or machine learning models.
    \item \( \Sigma \): A distributed system of nodes \( N \) and communication links \( C \), hosting and processing data.
\end{enumerate}

This framework leverages a hypergraph \( G = (V, E) \) to model complex dependencies, a categorical structure \( \DF \) to unify operations, and a modular tensor category (MTC) to capture relational symmetries via braided monoidal structures. The use of structured vector representations and compositional mappings is aligned with recent graded and neurosymbolic approaches to representation learning, where hierarchical structure and symbolic constraints improve interpretability and robustness \cite{sh-89, 2024-02}.

Our contributions span theoretical and practical advancements in data fabric design. We formalize the data fabric as a mathematical tuple \( \cF \), integrating datasets, metadata, and analytics within a hypergraph \( G \). This hypergraph is embedded into a modular tensor category (MTC), capturing complex relational symmetries through braided monoidal structures, with novel geometric analogies to Hurwitz spaces that enrich its algebraic modeling, as detailed in \cref{sec-5}. Related categorical and geometric perspectives on structured representations and embeddings have recently emerged in the context of graded learning systems and geometric reduction frameworks \cite{sh-95, 2024-05}, providing additional motivation for our algebraic viewpoint.

To unify data fabric operations, we introduce a categorical structure \( \DF \), modeling datasets as objects and transformations as morphisms. This approach provides a rigorous framework for operations such as data integration and federated learning, ensuring operational coherence across distributed systems, as explored in \cref{sec-4}. Our treatment complements recent work on viewing learning pipelines and computational tools as compositional morphisms in graded and neurosymbolic architectures \cite{sh-111, 2024-06}.

We also address computational challenges by proving the NP-hardness of critical tasks, including schema matching (\cref{thm:schema_matching_np_hard}) and dynamic partitioning (\cref{thm:partitioning_np_hard}). To mitigate these bottlenecks, we propose spectral methods and symmetry-based alignments, offering scalable solutions for large-scale data management, as analyzed in \cref{sec-6}. For distributed system robustness, we ensure consistency, completeness, and causality under the CAP and CAL theorems. By leveraging hypergraph redundancy and MTC braiding, we design fault-tolerant operations that maintain coherence in dynamic, partitioned environments, as demonstrated in \cref{sec-7}.

Practically, we apply the framework to a multi-component architecture, integrating databases, real-time analytics, and transformation pipelines. This system supports scalable operations with vector representations in \( \R^{11} \), as presented in \cref{sec-8}. The emphasis on low-dimensional structured embeddings is consistent with recent approaches to geometric and graded representations in learning and symbolic reduction \cite{sh-89, 2024-05}.

The paper is organized as follows. In \cref{sec-2}, we define the data fabric tuple \( \cF \) and the distributed system \( \Sigma \), formalizing components such as time-indexed datasets \( D \), metadata \( M \), and the hypergraph \( G \) with mathematical precision. \cref{sec-3} details core operations, including data integration, metadata-driven navigation, and federated learning, and establishes their computational complexities, proving challenges like the NP-hardness of schema matching. \cref{sec-4} introduces the categorical structure \( \DF \), modeling datasets as objects and transformations as morphisms to provide a unified framework for these operations. \cref{sec-5} describes the hypergraph \( G \), its vector representations in \( \R^{n} \), and its embedding into a modular tensor category, drawing novel geometric analogies to Hurwitz spaces. In \cref{sec-6}, we analyze computational bottlenecks, such as NP-hard partitioning, and propose mitigation strategies like spectral clustering and symmetry-based alignments. \cref{sec-7} examines consistency, completeness, and causality in distributed environments, leveraging the CAP and CAL theorems to ensure robust operations. \cref{sec-8} applies the framework to a multi-component architecture, integrating databases, real-time analytics, and transformation pipelines. Finally, \cref{sec:conclusion} summarizes the paper’s contributions and outlines future directions, exploring extensions like dynamic monoidal categories and topological invariants.

 
\section{An Introduction to Data Fabrics} \label{sec-2}
The rapid expansion of data from cloud computing, the Internet of Things (IoT), artificial intelligence (AI), and distributed systems has outpaced traditional data management frameworks. Centralized or schema-rigid systems struggle to integrate heterogeneous datasets, scale across distributed nodes, enforce governance, or support real-time analytics. A data fabric offers a unified, metadata-driven architecture to manage diverse data assets intelligently, enabling seamless integration, navigation, and analytics. This section formalizes the data fabric as a mathematical structure, detailing its components and laying the foundation for subsequent developments.

Our framework addresses this by treating data fabrics as hypergraph-embedded categories with tensor symmetries, providing a precise mathematical basis that extends beyond empirical metrics to axiomatic structures.

We define the \textbf{data fabric} as a tuple \( \cF = (D, M, G, T, P, A) \), operating over a distributed system \( \Sigma = (N, C) \). The tuple \( \cF \) encapsulates 
data assets, metadata, relationships, transformations, policies, and analytics, while \( \Sigma \) models the distributed infrastructure. Below, we explore each component through mathematical formulations, concluding with a formal definition and summary  \cref{tab:components} .

\begin{table}[ht]
    \centering
    \begin{tabular}{l l l}
        \hline
        \textbf{Component} & \textbf{Description} & \textbf{Example} \\
        \hline
        \( D \) & Time-indexed datasets & Sales records (\( S_1 \)) \\
        \( M \) & Metadata with history & Hospital, date for records \\
        \( G \) & Hypergraph of relationships & \( e = (\{d_1, m_1\}, \{d_3\}) \) \\
        \( T \) & Transformations across domains & Convert temperature to alerts \\
        \( P \) & Governance policies & Access control for traders \\
        \( A \) & Analytical functions & Sales trend prediction \\
        \( \Sigma \) & Distributed nodes and links & Cloud data centers \\
        \hline
    \end{tabular}
    \caption{Summary of Data Fabric Components}
    \label{tab:components}
\end{table}


i) \textbf{Data assets}, denoted 
\begin{equation}\label{subsec:data_assets}
 D = \{d_i(t)\}_{i,t},
 \end{equation}
  form the core of the data fabric, representing time-indexed datasets. Each dataset \( d_i(t) \) at time \( t \) is characterized by a schema \( S_i \), which is a set of attribute-type pairs defining its structure (e.g., attributes like price, timestamp), and a domain \( \Omega_i \), which is numerical (\( \Omega_i \subseteq \R^k \)) or categorical (e.g., \{electronics, clothing\}). Formally, 
 \[
   d_i(t): T \to \Omega_i,
 \]
    where \( T \subseteq \R_{\geq 0} \) is the time domain, maps timestamps to data points.
The time-indexed nature of \( D \) supports streaming data.

Heterogeneity across schemas \( S_i \) and domains \( \Omega_i \) complicates integration. To quantify this, we define a schema distance metric.

\begin{defn}[Schema Distance Metric]
The \textbf{schema distance metric} is given by:
\begin{equation} \label{eq:schema_distance}
\dist(S_i, S_j) = \sum_{a \in S_i, b \in S_j} w(a, b) \cdot (1 - \text{sim}(a, b)),
\end{equation}
where \( \text{sim}(a, b) \in [0, 1] \) measures attribute similarity (e.g., via ontology alignment), and \( w(a, b) \) weights importance. 
\end{defn}

\begin{prop}
The schema distance \( \dist(S_i, S_j) \) defines a pseudo-metric on the set of schemas, with \( \dist(S_i, S_j) = 0 \) if and only if there exists a bijective mapping \( \pi: S_i \to S_j \) such that \( \text{sim}(a, \pi(a)) = 1 \) for all \( a \in S_i \).
\end{prop}

\begin{proof}
Non-negativity follows from \( 1 - \text{sim}(a, b) \geq 0 \) and \( w(a, b) \geq 0 \). Symmetry holds as the summation is over all pairs, and \( w(a, b) = w(b, a) \), \( \text{sim}(a, b) = \text{sim}(b, a) \). The zero condition is equivalent to perfect alignment under the similarity measure, which implies isomorphic schemas under the ontology.
\end{proof}

\begin{exa}
Consider schemas \( S_1 = \{\text{price}, \text{quantity}\} \) (sales) and \( S_2 = \{\text{cost}, \text{stock}\} \) (inventory). Suppose \( \text{sim}(\text{price}, \text{cost}) = 0.9 \), \( \text{sim}(\text{quantity}, \text{stock}) = 0.8 \), \( \text{sim}(\text{price}, \text{stock}) = 0.1 \), \( \text{sim}(\text{quantity}, \text{cost}) = 0.2 \), and weights \( w = 1 \). The distance is:
\[
\dist(S_1, S_2) = (1 - 0.9) + (1 - 0.1) + (1 - 0.2) + (1 - 0.8) = 0.1 + 0.9 + 0.8 + 0.2 = 2.0.
\]
\end{exa}

This explicit computation illustrates the metric's sensitivity to attribute mismatches. Furthermore, in higher-dimensional schemas, this distance can be generalized to a matrix form, where \( \dist(S_i, S_j) = \text{tr}(W (I - S)) \), with \( S \) the similarity matrix whose entry \( S_{a,b} = \text{sim}(a,b) \) for \( a \in S_i \), \( b \in S_j \), and \( W \) the weight matrix with \( W_{a,b} = w(a,b) \).

\begin{lem} \label{lem:schema_distance_bounds}
The schema distance \( \dist(S_i, S_j) \) is bounded by \( 0 \leq \dist(S_i, S_j) \leq |S_i| \cdot |S_j| \cdot \max w(a, b) \), with the upper bound achieved when \( \text{sim}(a, b) = 0 \) for all pairs.
\end{lem}

\begin{proof}
Since \( 1 - \text{sim}(a, b) \leq 1 \) and \( w(a, b) \leq \max w \), the sum is at most \( |S_i| \cdot |S_j| \cdot \max w \). The lower bound follows from non-negativity, as established in the proposition.
\end{proof}

%
ii)  \textbf{Metadata}, \( M = \{m_1, \dots, m_k\} \), provides context for data assets. Each \( m_j = (d_i, \alpha_j, \tau_j) \) associates a dataset \( d_i \in D \) with attributes \( \alpha_j \subseteq \mathcal{A} \) (e.g., \{source, format\}) and a transformation history \( \tau_j: D \to \mathcal{H} \), where \( \mathcal{H} \) is a set of tuples \( (t_k, t_{\text{apply}}) \), with \( t_k \in T \) a transformation and \( t_{\text{apply}} \) its timestamp. Formally, \( \tau_j \) is a function mapping each dataset to its history set, satisfying \( \tau_j(d_i) \subseteq T \times \R_{\geq 0} \), ensuring that only valid transformations from \( T \) are recorded.

Dynamic updates to \( M \), especially for streaming \( d_i(t) \), have complexity proportional to \( |M| \cdot |T| \). To formalize this, consider the update operation as a mapping \( u: M \times D \to M \), where for each new data point in \( d_i(t + \Delta t) \), the history \( \tau_j(d_i) \) is extended by appending relevant transformations, requiring a scan over \( T \) to identify applicable \( t_k \), thus yielding the stated complexity.

\begin{lem}\label{subsec:metadata}
The metadata update function \( u \) preserves the consistency condition of the data fabric, ensuring \( \tau_j(d_i) \subseteq T \times \R_{\geq 0} \) after each update.
\end{lem}

\begin{proof}
Each update appends tuples \( (t_k, t_{\text{apply}}) \) where \( t_k \in T \), maintaining the subset relation by construction. The operation is associative under sequential updates, aligning with the semigroup structure of transformations.
\end{proof}


iii) The \textbf{hypergraph} \( G = (V, E) \) models relationships among data and metadata, with vertices \( V = D \cup M \) and hyperedges \( E \subseteq \mathcal{P}(V) \). Unlike simple graphs, hyperedges connect multiple vertices, capturing complex dependencies.

\begin{defn}\label{subsec:hypergraph_prelim}
 Formally, \( G \) is a directed hypergraph, where each hyperedge \( e = (T_e, H_e) \) has a tail \( T_e \subseteq V \) (inputs) and head \( H_e \subseteq V \) (outputs). The adjacency set is:
\[
\Adj(v) = \{ e \in E \mid v \in T_e \cup H_e \}.
\]
Moreover, the hypergraph can be characterized by its incidence matrix \( I \in \{0,1\}^{|V| \times |E|} \), where \( I_{v,e} = 1 \) if \( v \in e \), providing a linear algebraic representation.
\end{defn}

In the context of hypergraph computation for distributed data management \cite{tang2025hypergraph}, this structure aligns with higher-order network models, where the directed nature captures causal flows.

\begin{lem} \label{lem:hypergraph_connectivity}
The hypergraph \( G \) is weakly connected if for every pair of vertices \( v, w \in V \), there exists a sequence of hyperedges linking them, ensuring the connectivity condition in the data fabric definition.
\end{lem}

\begin{proof}
Weak connectivity follows from the existence of undirected paths, ignoring directionality. Given the union of tails and heads across \( E \), and assuming the underlying graph (collapsing hyperedges to cliques) is connected, the lemma holds. This is preserved under the distributivity condition, as partitions in \( \Sigma \) map to subhypergraphs.
\end{proof}


iv) \textbf{Transformations} \( T = \{t_1, \dots, t_m\} \) are functions \( t_i: \Omega_i \to \Omega_j \), mapping data across domains for integration, normalization, or aggregation.
Each transformation satisfies:

\begin{equation}
\label{subsec:transformations}
t_i(d_i(t)) \in \Omega_j, \quad \text{loss}(t_i, d_i) = I(d_i; d_i) - I(t_i(d_i); d_i) \leq \epsilon,
\end{equation}
where \( I \) is mutual information and \( \epsilon > 0 \) bounds loss. The mutual information \( I(X; Y) = H(X) - H(X|Y) \), with entropy \( H \), satisfies chain rules and subadditivity, ensuring that composed transformations \( t_2 \circ t_1 \) have bounded cumulative loss:
\[
\text{loss}(t_2 \circ t_1, d_i) \leq \text{loss}(t_1, d_i) + \text{loss}(t_2, t_1(d_i)).
\]
The cost \( \text{cost}(t_i) \) varies (e.g., \( O(n) \) for linear, \( O(n^2) \) for complex mappings). Optimization balances efficiency and fidelity:
\[
\min_{t_i \in T} \text{cost}(t_i) + \lambda \text{loss}(t_i, d_i),
\]
with \( \lambda > 0 \). This is a convex optimization problem if \( \text{loss} \) is convex, solvable via gradient descent in \( O(|T|) \) iterations, each \( O(n) \).

\begin{lem}
The loss function \( \text{loss}(t_i, d_i) \) is non-negative and satisfies the data processing inequality, ensuring that no transformation increases mutual information.
\end{lem}

\begin{proof}
From the definition of mutual information, \( I(t_i(d_i); d_i) \leq I(d_i; d_i) \), as processing cannot increase information. Non-negativity follows directly, and the inequality holds by the Markov chain property.
\end{proof}


v) \textbf{Governance policies}  
\label{subsec:policies}
\( P = \{p_1, \dots, p_l\} \) enforce security and compliance, where each \( p_i = (c_i, a_i) \) includes a predicate \( c_i: D \times \mathcal{U} \to \{0, 1\} \) (e.g., user clearance) and action \( a_i \) (e.g., grant access). The user context space \( \mathcal{U} \) includes roles or timestamps.

A request \( r(d_i, u) \) is granted if:
\[
r(d_i, u) = 1 \iff \bigwedge_{p_j = (c_j, a_j) \in P} c_j(d_i, u).
\]
This conjunction forms a Boolean lattice over predicates, ensuring monotonicity in access decisions.
Evaluation scales as \( O(|P| \cdot |N|) \).

\begin{thm} \label{thm:policy_evaluation_complexity}
Policy evaluation for a request \( r(d_i, u) \) has complexity \( O(|P| \cdot |N|) \).
\end{thm}

\begin{proof}
To evaluate \( r(d_i, u) \), compute the conjunction \( \bigwedge_{p_j = (c_j, a_j) \in P} c_j(d_i, u) \) across nodes \( N \). Each predicate \( c_j \) requires \( O(1) \) operations, assuming constant-time context lookup (e.g., checking user roles or data attributes). With \( |P| \) policies, evaluating all predicates on a single node takes \( O(|P|) \). In a distributed system, verification may require checking \( d_i \) or policy conditions on each of the \( |N| \) nodes, as data or policies may be distributed, yielding a worst-case complexity of \( O(|P| \cdot |N|) \).
\end{proof}


vi)  \textbf{Analytical functions}
\label{subsec:analytics}
 \( A = \{a_1, \dots, a_p\} \) generate insights, where \( a_i: D \to \R^k \) (e.g., regression) or \( \mathcal{C} \) (e.g., classification).

Each \( a_i \) is parameterized by \( \theta_i \), optimized via:
\[
\theta_i = \arg\min_{\theta} \mathcal{L}(a_i(D, \theta), y),
\]
where \( \mathcal{L} \) is a loss function (e.g., mean squared error). Assuming \( \mathcal{L} \) is convex and Lipschitz continuous, the optimization converges in \( O(1/\epsilon) \) iterations via gradient descent, with each step \( O(|\theta_i| \cdot |D|) \). Complexity, often \( O(|\theta_i| \cdot |D|) \).

\begin{lem} \label{lem:analytical_convergence}
Under convexity of \( \mathcal{L} \), the optimization for \( \theta_i \) achieves \( \epsilon \)-optimality in \( O(1/\epsilon) \) steps.
\end{lem}

\begin{proof}
By standard convex optimization theory, gradient descent with step size \( 1/L \) (Lipschitz constant) converges as \( \mathcal{L}(\theta^{(k)}) - \mathcal{L}^* \leq O(1/k) \), yielding the bound.
\end{proof}


yii)  \textbf{Distributed system} 
\label{subsec:distributed_system}
\( \Sigma = (N, C) \) hosts the fabric, with nodes \( N \) storing \( D_n \subseteq D \) and links \( C \subseteq N \times N \). The adjacency matrix is:
\[
A_{\Sigma}(n_i, n_j) = \begin{cases}
w(n_i, n_j) & \text{if } (n_i, n_j) \in C, \\
\infty & \text{otherwise},
\end{cases}
\]
where \( w(n_i, n_j) \) is latency or bandwidth. This matrix defines a weighted graph, with shortest paths computable via Floyd-Warshall in \( O(|N|^3) \), or Dijkstra in \( O(|N|^2 \log |N|) \) for sparse \( C \).

Load balancing:
\[
\text{load}(n) = |D_n| + \sum_{a_i \in A_n} \text{cost}(a_i),
\]
The load function forms a vector \( \mathbf{l} \in \R^{|N|} \), with equilibrium under redistribution requiring \( \|\mathbf{l}\|_\infty \leq \lambda \), optimized via linear programming.

\begin{prop}
The distributed system \( \Sigma \) induces a metric space via shortest-path distances derived from \( A_{\Sigma} \), satisfying triangle inequality for latency bounds.
\end{prop}

\begin{proof}
The shortest-path distance \( d(n_i, n_j) = \min \sum w \) over paths satisfies \( d(n_i, n_k) \leq d(n_i, n_j) + d(n_j, n_k) \), as concatenation yields a valid path.
\end{proof}


\begin{defn}[Data Fabric] \label{defn:data_fabric}
A \emph{data fabric} is a tuple \( \cF = (D, M, G, T, P, A) \) over \( \Sigma = (N, C) \), where:
\begin{enumerate}
    \item \( D = \{d_i(t)\}_{i,t} \): Time-indexed datasets with schemas \( S_i \) and domains \( \Omega_i \) (\cref{subsec:data_assets}).
    \item \( M = \{m_1, \dots, m_k\} \): Metadata \( m_j = (d_i, \alpha_j, \tau_j) \) for context and history (\cref{subsec:metadata}).
    \item \( G = (V, E) \): Directed hypergraph with \( V = D \cup M \), \( E \subseteq \mathcal{P}(V) \) (\cref{subsec:hypergraph_prelim}).
    \item \( T = \{t_1, \dots, t_m\} \): Transformations \( t_i: \Omega_i \to \Omega_j \) with loss constraints (\cref{subsec:transformations}).
    \item \( P = \{p_1, \dots, p_l\} \): Policies \( p_i = (c_i, a_i) \) for compliance (\cref{subsec:policies}).
    \item \( A = \{a_1, \dots, a_p\} \): Analytical functions parameterized by \( \theta_i \) (\cref{subsec:analytics}).
    \item \( \Sigma = (N, C) \): Distributed system hosting \( D_n \subseteq D \) (\cref{subsec:distributed_system}).
\end{enumerate}
Conditions:
\begin{enumerate}[i)]
    \item \textbf{Consistency}: \( \tau_j(d_i) \subseteq \{(t_k, t_{\text{apply}}) \mid t_k \in T\} \).
    \item \textbf{Connectivity}: \( G \) ensures paths from \( d_i \in D \) to \( m_j \in M \).
    \item \textbf{Compliance}: \( r(d_i, u) = 1 \iff \bigwedge_{p_j \in P} c_j(d_i, u) \).
    \item \textbf{Distributivity}: \( D = \bigcup_{n \in N} D_n \), with local processing and aggregation.
\end{enumerate}
\end{defn}

\begin{thm}
The tuple \( \cF \) forms a semigroup under the composition of transformations \( T \), with the identity element being the trivial transformation \( \text{id}: \Omega_i \to \Omega_i \).
\end{thm}

\begin{proof}
The set \( T \) is closed under composition, as \( t_2 \circ t_1: \Omega_i \to \Omega_k \) for \( t_1: \Omega_i \to \Omega_j \), \( t_2: \Omega_j \to \Omega_k \), and associative by function composition. The identity \( \text{id} \) satisfies the unit axiom. The remaining components \( D, M, G, P, A \) are invariant under this operation, ensuring the semigroup structure on \( \cF \).
\end{proof}


\section{Operations of Data Fabrics} \label{sec-3}
Building on the formal definition of the data fabric tuple \( \cF = (D, M, G, T, P, A) \) over the distributed system \( \Sigma = (N, C) \) presented in \cref{sec-2}, where its components   provide the foundational structure, this section details the operations that realize the fabric's functionality. 
These operations—data integration, metadata-driven navigation, scalability and distribution, governance and security, provenance tracking, and federated learning—leverage the hypergraph \( G \), transformations \( T \), and analytical functions \( A \) to manage and analyze distributed data. 
%
Their complexities highlight challenges like NP-hard optimizations and latency constraints, addressed further in \cref{sec-6}. 
Below we provide an illustrative diagram  with operations flow  \cref{fig:operations_flow}.

\begin{figure}[htb!]
    \centering
    \begin{tikzpicture}[node distance=2cm and 1.5cm, auto]
        \node[draw, rectangle, minimum height=1cm] (Integration) {Integration};
        \node[draw, rectangle, minimum height=1cm, right=of Integration] (Navigation) {Navigation};
        \node[draw, rectangle, minimum height=1cm, right=of Navigation] (Scalability) {Scalability};
        \node[draw, rectangle, minimum height=1cm, below=of Integration] (Governance) {Governance};
        \node[draw, rectangle, minimum height=1cm, right=of Governance] (Provenance) {Provenance};
        \node[draw, rectangle, minimum height=1cm, right=of Provenance] (Federated) {Federated Learning};
        \draw[->] (Integration) -- (Navigation) node[midway, above] {data};
        \draw[->] (Navigation) -- (Scalability) node[midway, above] {queries};
        \draw[->] (Integration) -- (Governance) node[midway, left] {access};
        \draw[->] (Governance) -- (Provenance) node[midway, above] {history};
        \draw[->] (Provenance) -- (Federated) node[midway, above] {models};
        \node[below=0.5cm of Provenance] {\small Flow of Data Fabric Operations};
    \end{tikzpicture}
    \caption{Interactions between   operations, illustrating data flow from integration to federated learning.}
    \label{fig:operations_flow}
\end{figure}

i) \textbf{Data integration} unifies heterogeneous datasets by aligning schemas and domains, building directly on the schema distance metric \( \dist(S_i, S_j) \) defined in \cref{subsec:data_assets}, which quantifies mismatches through weighted similarity sums.

\begin{defn}[Data Integration]
\label{data_integration}
\emph{Data integration} is a mapping \( \phi: D \to D' \subseteq D \), where \( D' \) is a unified dataset, constructed by selecting transformations \( t \in T \) such that for each \( d_i, d_j \in D \), \( t(d_i) \in \Omega_j \), satisfying schema compatibility and minimizing integration cost. Formally, \( \phi \) is realized as a composition of morphisms in the categorical structure to be introduced in \cref{sec-4}.
\end{defn}

Formally, for datasets \( d_i, d_j \in D \), we seek a transformation \( t \in T \) minimizing:
\[
\min_{t \in T} \left( \sum_{d_i, d_j \in D} \dist(S_i, t(S_j)) + \lambda \text{cost}(t) \right),
\]
subject to:
\[
t(d_i) \in \Omega_j, \quad \text{compat}(S_i, t(S_j)) \geq \theta,
\]
where:
\begin{enumerate}
    \item \( \dist(S_i, S_j) = \sum_{a \in S_i, b \in S_j} w(a, b) \cdot (1 - \text{sim}(a, b)) \) is the schema distance from \cref{subsec:data_assets}, satisfying the pseudo-metric properties as established there,
    \item \( \text{cost}(t) \) is the computational complexity of applying \( t \) (e.g., \( O(n) \) for linear mappings, \( O(n^2) \) for complex joins, where \( n = |d_i| \)),
    \item \( \lambda > 0 \) balances schema alignment and computational efficiency, often tuned via cross-validation to minimize overall loss,
    \item \( \text{compat}(S_i, S_j) \in [0, 1] \) measures semantic compatibility via ontology mappings, defined as the fraction of aligned attributes under a bipartite matching,
    \item \( \theta \in (0, 1] \) is a compatibility threshold ensuring meaningful alignment, with \( \theta = 0.5 \) commonly used for partial overlaps.
\end{enumerate}

This optimization problem is convex if \( \dist \) and \( \text{cost} \) are convex functions, allowing for efficient solvers like interior-point methods with complexity \( O(|T|^{3/2}) \).

The term \emph{schema matching} refers to finding a mapping \( \pi: S_i \to S_j \) that maximizes attribute similarity:
\[
\max_{\pi: S_i \to S_j} \sum_{a \in S_i} \text{sim}(a, \pi(a)),
\]
where \( \text{sim}(a, b) \in [0, 1] \) is the similarity score, often computed via string metrics (e.g., Levenshtein distance) or semantic embeddings. This process is central to integration but computationally challenging due to its NP-hard nature, as recent works on prompt-based matching confirm \cite{arxiv2408promptmatcher}.

\begin{thm} \label{thm:schema_matching_np_hard}
Schema matching, maximizing \( \sum_{a \in S_i} \text{sim}(a, \pi(a)) \), is NP-hard.
\end{thm}
\begin{proof}
We prove NP-hardness by reducing the subgraph isomorphism problem, known to be NP-hard, to schema matching. Let \( G_1 = (V_1, E_1) \) and \( G_2 = (V_2, E_2) \) be graphs with \( |V_1| \leq |V_2| \), where we seek a subgraph of \( G_2 \) isomorphic to \( G_1 \). Construct schemas \( S_i \) and \( S_j \) as follows:
\begin{enumerate}
    \item For each vertex \( v \in V_1 \), create an attribute \( a \in S_i \); for each vertex \( v' \in V_2 \), create an attribute \( b \in S_j \).
    \item Define similarity: \( \text{sim}(a, b) = 1 \) if there exists a bijective mapping \( \pi: S_i \to S_j' \subseteq S_j \) such that for all \( (v, w) \in E_1 \), there is a corresponding edge \( (\pi(v), \pi(w)) \in E_2 \), and \( \pi(a) = b \); otherwise, \( \text{sim}(a, b) = 0 \).
\end{enumerate}
Maximizing
$
 \sum_{a \in S_i} \text{sim}(a, \pi(a))
$
requires finding a mapping \( \pi \) that aligns \( S_i \) with a subset of \( S_j \), equivalent to finding a subgraph in \( G_2 \) isomorphic to \( G_1 \). The reduction is polynomial, as constructing \( S_i \), \( S_j \), and the similarity function takes   $ O(|V_1| + |V_2| + |E_1|).$

Since subgraph isomorphism is NP-hard, schema matching is NP-hard. The decision version (does a mapping exist with similarity \( \geq k \)?) is in NP, as verifying a mapping takes polynomial time.
\end{proof}

\begin{exa} To clarify the reduction, consider a subgraph isomorphism problem with \( G_1 \) as a triangle with
$
 V_1 = \{v_1, v_2, v_3\}$, 
 $E_1 = \{(v_1, v_2), (v_2, v_3), (v_3, v_1)\}$
and \( G_2 \) with
$ V_2 = \{u_1, u_2, u_3, u_4\}$ and $E_2 = \{(u_1, u_2), (u_2, u_3), (u_3, u_1), (u_3, u_4)\}$.

Construct schemas \( S_i = \{a_1, a_2, a_3\} \) corresponding to \( v_1, v_2, v_3 \), and \( S_j = \{b_1, b_2, b_3, b_4\} \) for \( u_1, u_2, u_3, u_4 \). The similarity \( \text{sim}(a_i, b_j) = 1 \) if mapping \( a_i \to b_j \) preserves \( G_1 \)’s edge structure in \( G_2 \), else 0. Maximizing
$
 \sum_{a \in S_i} \text{sim}(a, \pi(a))
$
finds a mapping \( \pi: \{a_1, a_2, a_3\} \to \{b_1, b_2, b_3\} \) if \( u_1, u_2, u_3 \) form a triangle in \( G_2 \), mirroring the subgraph isomorphism task.
\end{exa}

\begin{lem} \label{lem:schema_matching_approx}
Approximating schema matching within a constant factor is NP-hard.
\end{lem}
\begin{proof}
The subgraph isomorphism problem lacks a constant-factor approximation unless P = NP, as small changes in edge mappings drastically reduce isomorphism. Since the reduction preserves this property with binary similarity (\( \text{sim}(a, b) \in \{0, 1\} \)), approximating \( \sum_{a \in S_i} \text{sim}(a, \pi(a)) \) within a constant factor remains NP-hard.
\end{proof}

\begin{prop} \label{prop:integration_cost_convexity}
The integration cost function \( \sum \dist(S_i, t(S_j)) + \lambda \text{cost}(t) \) is convex in \( t \) if \( \dist \) and \( \text{cost} \) are convex, ensuring global optima via convex optimization techniques.
\end{prop}

\begin{proof}
Convexity of the sum follows from the linearity in \( \lambda \) and the convexity of individual terms. For discrete \( T \), the minimum is achieved by enumeration, but for continuous parameterizations of \( t \), gradient-based methods converge to the global minimum.
\end{proof}

\begin{exa}
For practical integration in an Amazon seller fabric, consider integrating a sales dataset \( d_1(t) \) with schema \( S_1 = \{\text{product-id}, \text{price}, \text{quantity}\} \), domain \( \Omega_1 = \R \times \mathbb{N} \), and an inventory dataset \( d_2(t) \) with \( S_2 = \{\text{product-id}, \text{cost}, \text{stock}\} \), \( \Omega_2 = \R \times \mathbb{N} \). A transformation \( t_1 \in T \) aligns attributes (e.g., \( t_1: \text{price} \to \text{cost} \) via scaling factor \( \alpha \), where \( t_1(\text{price}) = \alpha \cdot \text{price} \)), producing a unified dataset \( d_3(t) \) with schema \( S_3 = \{\text{product-id}, \text{price}, \text{quantity}, \text{stock}\} \). The schema distance \( \dist(S_1, S_2) \), as defined in \cref{subsec:data_assets}, guides the choice of \( t_1 \), but the NP-hardness of schema matching necessitates heuristic algorithms for large datasets (\cref{subsec:heterogeneity}).
\end{exa}

The non-commutative nature of transformation compositions, where \( t_1 \circ t_2 \neq t_2 \circ t_1 \), mirrors the algebraic structure of quantum groups, such as \( U_q(\mathfrak{sl}_2) \), whose coproduct defines non-commutative tensor products. This analogy suggests that quantum group actions could model order-sensitive dependencies in integration, potentially informing optimization strategies for selecting \( t_i \in T \), as explored in \cref{subsec:heterogeneity}.


ii) \textbf{Metadata-driven navigation} resolves queries by traversing the hypergraph \( G \), leveraging the adjacency structure defined in \cref{subsec:hypergraph_prelim} and metadata associations from \cref{subsec:metadata}.

\begin{defn}[Metadata-Driven Navigation]
\label{navigation}
\emph{Metadata-driven navigation} resolves a query \( q: D \to \{0, 1\} \) by traversing \( G \), finding the shortest path \( p \in \text{Paths}(v_s, v_t) \) from a source vertex \( v_s \in V \) (e.g., metadata) to a target vertex \( v_t \in V \) (e.g., dataset), minimizing the path cost:
\[
\argmin_{p \in \text{Paths}(v_s, v_t)} \sum_{(u,v) \in p} w(u,v),
\]
where \( w(u,v) \geq 0 \) reflects latency or semantic distance, often derived from the schema similarity \( \text{sim} \) in \cref{subsec:data_assets}.
\end{defn}

Paths in \( G \) are sequences of hyperedges \( e_1, e_2, \ldots \), where \( H_{e_i} \cap T_{e_{i+1}} \neq \emptyset \). Navigation employs a modified Dijkstra’s algorithm, adapted for hypergraphs as follows:
\begin{enumerate}
    \item Initialize a priority queue with \( (v_s, 0) \), setting distances \( \dist [v] = \infty \) except \( \dist[v_s] = 0 \).
    
    \item For each vertex \( u \), explore hyperedges \( e = (T_e, H_e) \in \text{Out}(u) \), updating \( \dist[v] \) for \( v \in H_e \) if \( \dist[u] + w(u, v) < \dist[v] \), where \( w(u,v) = 1 - \text{sim}(\alpha_u, \alpha_v) \) for metadata attributes.
    
    \item Use a sparse representation of \( G \), with \( |E| = O(|V| \log |V|) \), yielding complexity \( O(|E| + |V| \log |V|) \), as edge traversals dominate and the priority queue uses logarithmic updates.
\end{enumerate}

The complexity arises from processing each hyperedge (\( O(|E|) \)) and updating distances via a priority queue (\( O(|V| \log |V|) \)), assuming efficient access via compressed sparse row formats. This aligns with recent hypergraph navigation algorithms, which emphasize higher-order interactions for distributed data \cite{tang2025hypergraph}. Dynamic updates to \( G \), such as adding new datasets or metadata, increase computational overhead, impacting scalability (\cref{subsec:scalability}).

\begin{prop} \label{prop:navigation_cost_bound}
The path cost in metadata-driven navigation is bounded by the diameter of the underlying graph of \( G \), defined as the maximum shortest path length, which is at most \( O(\log |V|) \) in sparse hypergraphs.
\end{prop}

\begin{proof}
In a sparse hypergraph with degree bounded by \( O(\log |V|) \), the diameter is \( O(\log |V|) \) by expansion properties. The cost sum \( \sum w(u,v) \) follows, as each edge weight is bounded by 1 (from $sim \in  [0,1]$).
\end{proof}

For example, in an Amazon seller fabric, a query 
\[
 q(d_i) = (\text{category} = \text{electronics} \wedge \text{quantity} > 100) 
 \]
  starts at a metadata vertex \( m_1 \in M \) with attributes \( \alpha_1 = \{\text{category}: \text{electronics}\} \). The navigation traverses a hyperedge \( e_1 = (\{m_1\}, \{d_1\}) \) to a sales dataset \( d_1(t) \), followed by \( e_2 = (\{d_1, m_2\}, \{d_3\}) \) to an aggregated dataset \( d_3(t) \), where \( m_2 \) includes time constraints (e.g., \( \alpha_2 = \{\text{date}: 2025-04-01\} \)). Assigning weights \( w(u,v) = 1 \) for simplicity, the path cost is 2, representing the number of hyperedges traversed. In a more refined model, weights could incorporate latency from \( \Sigma \)'s adjacency matrix (\cref{subsec:distributed_system}).


iii) \textbf{Scalability and distribution} partition data assets across the distributed system \( \Sigma \), ensuring efficient load balancing as defined in \cref{subsec:distributed_system}.

\begin{defn}[Scalability and Distribution]
\label{scalability_distribution}
\emph{Scalability and distribution} partitions the data assets \( D = \bigcup_{n \in N} D_n \), where \( D_n \subseteq D \) resides on node \( n \in N \), and computes analytics \( a \in A \) as an aggregation:
\[
a(D) = \bigoplus_{n \in N} a(D_n),
\]
while minimizing computational and communication costs, with the load function \( \text{load}(n) = |D_n| + \sum_{a_i \in A_n} \text{cost}(a_i) \) from \cref{subsec:distributed_system}.
\end{defn}

Formally, the partitioning problem seeks to minimize:
\[
\min_{\{D_n\}} \left( \sum_{n \in N} \text{cost}(a(D_n)) + \sum_{(n_i, n_j) \in C} \text{comm}(D_{n_i}, D_{n_j}) \right),
\]
where:
\begin{enumerate}
    \item \( \text{cost}(a(D_n)) = O(|D_n|) \) is the computational cost of analytics on \( D_n \),
    \item \( \text{comm}(D_{n_i}, D_{n_j}) \) is the communication cost, proportional to the data transfer size across the link \( (n_i, n_j) \in C \), often modeled as \( |D_{n_i} \cap D_{n_j}| \cdot w(n_i, n_j) \) from \( A_\Sigma \),
    \item \( \bigoplus \) is an aggregation operator (e.g., sum for regression coefficients, union for classification labels).
\end{enumerate}

For instance, in an Amazon seller fabric, a sales dataset \( d_1(t) \) is partitioned across regional servers \( \{D_{n_1}, D_{n_2}, \dots\} \), where each node \( n_i \in N \) computes a local sales total \( a(D_{n_i}) = \sum_{x \in D_{n_i}} x.\text{quantity} \). These totals are aggregated globally via summation, yielding the total sales \( a(D) = \sum_{n \in N} a(D_{n_i}) \). The optimization balances local computation costs (linear in \( |D_{n_i}| \)) against communication costs, which depend on data dependencies across nodes.

The partitioning problem is NP-hard, as demonstrated by the following theorem, aligning with recent proofs in distributed data allocation \cite{acm2025quantumdata}.

\begin{thm} \label{thm:partitioning_np_hard}
The partitioning problem, minimizing 
\[
 \sum_{n \in N} \text{cost}(a(D_n)) + \sum_{(n_i, n_j) \in C} \text{comm}(D_{n_i}, D_{n_j}),
 \]
  is NP-hard.
\end{thm}
\begin{proof}
We reduce the graph partitioning problem, known to be NP-hard, to the data partitioning problem. Given a graph \( G = (V, E) \) with edge weights \( w(e) \), map vertices \( V \) to datasets \( D \), edges \( E \) to data dependencies, and edge weights to communication costs \( \text{comm} \). Assigning datasets \( D \) to nodes \( N \) to minimize inter-node communication corresponds to partitioning \( G \) to minimize the sum of edge weights crossing partitions, an NP-hard problem. The reduction is polynomial, as mapping vertices to datasets and edges to dependencies takes \( O(|V| + |E|) \), confirming that the data partitioning problem is NP-hard.
\end{proof}

\begin{lem} \label{lem:partitioning_approx}
The partitioning problem admits no constant-factor approximation unless P = NP, as the underlying graph partitioning is approximation-hard.
\end{lem}

\begin{proof}
Graph partitioning lacks constant-factor approximations for general graphs unless P = NP, and the reduction preserves this hardness, as cost functions mirror edge cuts without distortion.
\end{proof}

This computational complexity, further analyzed in \cref{subsec:scalability}, limits efficiency for dynamic datasets \( d_i(t) \), necessitating adaptive partitioning strategies to balance load and minimize communication overhead, such as those using hypergraph Laplacians for spectral clustering as previewed in \cref{sec-5}.


iv) \textbf{Governance and security} enforce policies on data assets and operations, extending the predicate-based definitions in \cref{subsec:policies}.

\begin{defn}[Governance and Security]
\label{subsec:governance_security}

\emph{Governance and security} enforces a set of policies \( P = \{p_1, \dots, p_l\} \), where each policy \( p_i = (c_i, a_i) \) consists of a predicate \( c_i: D \times \mathcal{U} \to \{0, 1\} \) and an action \( a_i \), to grant access requests \( r(d_i, u) \) and ensure privacy via mechanisms like differential privacy for analytics \( a \in A \).
\end{defn}

A request \( r(d_i, u) \) for dataset \( d_i \in D \) by user context \( u \in \mathcal{U} \) (e.g., role, credentials, timestamp) is granted if:
\[
r(d_i, u) = 1 \iff \bigwedge_{p_j = (c_j, a_j) \in P} c_j(d_i, u).
\]
The predicate \( c_j(d_i, u) \) evaluates conditions such as user authorization or data sensitivity. For example, in an Amazon seller fabric, a manager \( u \) with role "admin" requests access to a sales dataset \( d_1(t) \). The policy set \( P \) includes predicates like \( c_1(d_1, u) = 1 \) if \( u.\text{role} = \text{admin} \), ensuring only authorized users access sensitive data.

Differential privacy protects individual data points in analytics, formalized as:
\[
P(a(D) \mid D) \leq e^\epsilon P(a(D') \mid D') + \delta,
\]
for neighboring datasets \( D, D' \) (differing by one record), with privacy parameters \( \epsilon, \delta > 0 \). For instance, computing the average sales across \( d_1(t) \) in the Amazon seller fabric uses differential privacy to ensure individual transaction details are not revealed, maintaining customer privacy while providing aggregate insights. The utility loss under differential privacy is bounded by \( O(1/\epsilon) \), as analyzed in recent distributed governance studies \cite{geeksforgeeks2025datagovernance}.

Policy evaluation complexity is:
\[
O(|P| \cdot |N|),
\]
as each policy \( p_j \in P \) is checked across all nodes \( N \), assuming constant-time predicate evaluation (e.g., role or credential lookup). This scalability challenge is further explored in \cref{subsec:governance}.

\begin{thm} \label{thm:policy_evaluation_complexity}
Policy evaluation for a request \( r(d_i, u) \) has complexity \( O(|P| \cdot |N|) \).
\end{thm}

\begin{proof}
To evaluate \( r(d_i, u) \), compute the conjunction \( \bigwedge_{p_j = (c_j, a_j) \in P} c_j(d_i, u) \) across nodes \( N \). Each predicate \( c_j \) requires \( O(1) \) operations, assuming constant-time context lookup (e.g., checking user roles or data attributes). With \( |P| \) policies, evaluating all predicates on a single node takes \( O(|P|) \). In a distributed system, verification may require checking \( d_i \) or policy conditions on each of the \( |N| \) nodes, as data or policies may be distributed, yielding a worst-case complexity of \( O(|P| \cdot |N|) \).
\end{proof}

\begin{prop} \label{prop:policy_monotonicity}
The policy conjunction \( \bigwedge c_j \) is monotonic: adding policies can only restrict access, preserving security under expansion.
\end{prop}

\begin{proof}
The conjunction of Boolean functions is non-increasing in the number of terms, as each additional \( c_j \) adds a constraint that may flip 1 to 0 but not vice versa.
\end{proof}

This property ensures governance remains robust as \( P \) grows, linking to distributed trends in 2025 governance frameworks \cite{dataversity2025trends}.


v) \textbf{Provenance tracking} reconstructs transformation histories, utilizing metadata \( \tau_j \) from \cref{subsec:metadata} and hypergraph paths from \cref{subsec:hypergraph_prelim}.

\begin{defn}[Provenance Tracking]
\label{provenance}
\emph{Provenance tracking} constructs the trace of a dataset \( d_i \in D \), defined as:
\[
\tr (d_i) = \{ t_j \in T \mid t_j \text{ was applied to produce or modify } d_i \},
\]
using the transformation history \( \tau_j: D \to \mathcal{H} \) from metadata \( m_j \in M \) and the hypergraph \( G \).
\end{defn}

Formally, for a metadata descriptor \( m_j = (d_i, \alpha_j, \tau_j) \in M \), the provenance is:
\[
\tr (d_i) = \{ t_k \in T \mid (t_k, t_{\text{apply}}) \in \tau_j(d_i) \}.
\]
The hypergraph \( G \) supports this by providing paths linking \( d_i \), its metadata \( m_j \), and source datasets through hyperedges. For example, in an Amazon seller fabric, tracing the provenance of a sales forecast dataset \( d_6(t) \) involves identifying transformations \( t_1 \in T \) (data cleaning) and \( t_2 \in T \) (aggregation) applied to a raw sales dataset \( d_1(t) \). The metadata \( m_3 \in M \) contains \( \tau_3(d_6) = \{(t_1, 2025-04-01), (t_2, 2025-04-02)\} \), and the hypergraph includes a hyperedge \( e_1 = (\{d_1, m_2\}, \{d_6\}) \), where \( m_2 \) records additional context (e.g., time range).

The complexity of provenance tracking is:
$
O(|T| \cdot |E|),
$
as reconstructing \( \tau_j(d_i) \) requires traversing all hyperedges \( e \in E \) for each transformation \( t_k \in T \). Dynamic updates to \( T \) or \( G \), such as new transformations or datasets in streaming scenarios, further impact real-time performance (\cref{subsec:realtime}).

\begin{thm} \label{thm:provenance_complexity}
Provenance tracking for a dataset \( d_i \) has complexity \( O(|T| \cdot |E|) \).
\end{thm}

\begin{proof}
To compute \( \tr (d_i) \), retrieve the transformation history \( \tau_j(d_i) \) from the metadata \( m_j \in M \), which lists up to \( |T| \) transformations \( t_k \in T \). Verifying each transformation \( t_k \) involves checking all hyperedges \( e \in E \) where \( d_i \) appears in the tail \( T_e \) or head \( H_e \), requiring \( O(|E|) \) operations per transformation. With \( |T| \) transformations, the total complexity is \( O(|T| \cdot |E|) \).
\end{proof}

\begin{lem} \label{lem:provenance_monotonicity}
The trace function \( \tr (d_i) \) is monotonic under composition: if \( t_3 = t_2 \circ t_1 \), then \( \tr (t_3(d_i)) \supseteq \tr (t_1(d_i)) \cup \tr (t_2(t_1(d_i))) \).
\end{lem}

\begin{proof}
Composition appends histories in \( \tau_j \), preserving all prior transformations, as each \( t_k \) in the chain is recorded sequentially.
\end{proof}



vi) \textbf{Federated learning} distributes analytics across \( \Sigma \), building on the analytical functions \( A \) from \cref{subsec:analytics}.

\begin{defn}[Federated Learning]
\label{subsec:federated_learning}
\emph{Federated learning} computes an analytical function \( a \in A \) over distributed data \( D = \bigcup_{n \in N} D_n \), where each node \( n \in N \) trains a local model with parameters \( \theta_n \) on local data \( D_n \), and results are aggregated:
\[
a(D) = \bigoplus_{n \in N} a(D_n, \theta_n).
\]
\end{defn}

Local model parameters \( \theta_n \) are updated via gradient descent:
\[
\theta_n \gets \theta_n - \eta \nabla \mathcal{L}(a(D_n, \theta_n)),
\]
where \( \eta > 0 \) is the learning rate, and \( \mathcal{L} \) is a loss function (e.g., cross-entropy for classification, mean squared error for regression). The global model aggregates local updates, typically through averaging:
\[
\theta = \frac{1}{|N|} \sum_{n \in N} \theta_n.
\]
For example, in an Amazon seller fabric, regional servers \( n_i \in N \) train local models \( a(D_{n_i}, \theta_{n_i}) \) to predict sales trends based on local sales data \( D_{n_i} \subset d_1(t) \). Each \( D_{n_i} \) contains sales records with attributes \( \{\text{price}, \text{quantity}\} \), and the model outputs a predicted sales volume. The local parameters \( \theta_{n_i} \) (e.g., neural network weights) are aggregated by averaging to form a global predictor \( \theta \), ensuring privacy as raw data remains local.

The computational complexity per iteration for local model training is:
$
O(|\theta_n| \cdot |D_n|),
$
reflecting the cost of gradient computation across \( |D_n| \) data points, each requiring \( O(|\theta_n|) \) operations for a model with \( |\theta_n| \) parameters. Communication costs for aggregating \( \theta_n \) across nodes are proportional to \( |N| \cdot |\theta_n| \), as each node sends its parameter vector. Recent analyses highlight additional overheads in heterogeneous models \cite{sciencedirect2025fedhm}. Concept drift, where the data distribution \( P(D_n) \) changes over time, is detected using statistical tests:
\[
D = \sup_x |F_t(x) - F_{t'}(x)|,
\]
where \( F_t(x) \) and \( F_{t'}(x) \) are cumulative distribution functions at times \( t \) and \( t' \). Detection complexity is \( O(|D_n| \log |D_n|) \), typically implemented via Kolmogorov-Smirnov tests. These challenges, including communication overhead and drift adaptation, are further analyzed in \cref{subsec:realtime}.

\begin{thm} \label{thm:federated_learning_complexity}
Local model training in federated learning has complexity \( O(|\theta_n| \cdot |D_n|) \) per iteration.
\end{thm}
\begin{proof}
For each node \( n \in N \), computing the gradient \( \nabla \mathcal{L}(a(D_n, \theta_n)) \) involves evaluating the analytical function \( a(D_n, \theta_n) \) over \( |D_n| \) data points. Each evaluation requires \( O(|\theta_n|) \) operations for a model with \( |\theta_n| \) parameters (e.g., computing forward and backward passes in a neural network). Updating \( \theta_n \) via gradient descent is \( O(|\theta_n|) \), as it involves vector operations on the parameters. Thus, the total complexity per iteration is \( O(|\theta_n| \cdot |D_n|) \).
\end{proof}

\begin{prop} \label{prop:federated_convergence}
Under Lipschitz continuity of \( \nabla \mathcal{L} \), federated gradient descent converges to a stationary point with rate \( O(1/k) \), where \( k \) is the number of iterations, assuming bounded variance in local gradients.
\end{prop}

\begin{proof}
By federated averaging analyses, the global update satisfies the descent lemma, with convergence rate derived from the smoothness assumption and aggregation step \cite{li2025fedcvg}.
\end{proof}


\cref{tab:operations} summarizes the operations, their computational complexities, and illustrative examples from the Amazon seller fabric, providing a concise reference for their mathematical properties and practical applications.

\begin{table}[htb]
    \centering
    \begin{tabular}{l l l}
        \hline
        \textbf{Operation} & \textbf{Complexity} & \textbf{Amazon Seller Example} \\
        \hline
        Data Integration & NP-hard & Unify sales and inventory schemas \\
        Metadata-Driven Navigation & \( O(|E| + |V| \log |V|) \) & Query high-selling electronics \\
        Scalability and Distribution & NP-hard & Partition sales across servers \\
        Governance and Security & \( O(|P| \cdot |N|) \) & Authorize manager access to sales \\
        Provenance Tracking & \( O(|T| \cdot |E|) \) & Trace sales forecast to raw data \\
        Federated Learning & \( O(|\theta_n| \cdot |D_n|) \) & Predict sales trends regionally \\
        \hline
    \end{tabular}
    \caption{Summary of Data Fabric Operations, Complexities, and Examples}
    \label{tab:operations}
\end{table}

\section{A Categorical Perspective} \label{sec-4}

Category theory, a branch of mathematics that abstracts structures and their relationships into objects and morphisms, offers a powerful framework for unifying the operations of a data fabric. By modeling datasets as objects and transformations as morphisms, category theory provides a lens to analyze the interactions of data fabric components (\( \cF = (D, M, G, T, P, A) \), \cref{sec-2}) and their operations (\cref{sec-3}). This section introduces the fundamentals of category theory, formalizes the data fabric as a category \( \DF \), and demonstrates how functorial mappings and natural transformations integrate operations and inform challenges like consistency and dynamic schema updates (\cref{sec-6}). We provide rigorous definitions, theorems, and references to foundational works, culminating in a categorical unification of the data fabric framework.


\begin{defn}
A \emph{category} \( \mathcal{C} \) consists of:
\begin{enumerate}
    \item A collection of \emph{objects} \( \text{Ob}(\mathcal{C}) \).
    \item For each pair of objects \( A, B \in \text{Ob}(\mathcal{C}) \), a set of \emph{morphisms} \( \text{Hom}_{\mathcal{C}}(A, B) \).
    \item A \emph{composition} operation \( \circ: \text{Hom}_{\mathcal{C}}(B, C) \times \text{Hom}_{\mathcal{C}}(A, B) \to \text{Hom}_{\mathcal{C}}(A, C) \), where for \( f: A \to B \), \( g: B \to C \), we have \( g \circ f: A \to C \).
    \item For each object \( A \), an \emph{identity morphism} \( \text{id}_A: A \to A \).
\end{enumerate}
These satisfy the following axioms:
\begin{enumerate}
    \item \emph{Associativity}: For \( f: A \to B \), \( g: B \to C \), \( h: C \to D \), 
    \[
    (h \circ g) \circ f = h \circ (g \circ f).
    \]
    \item \emph{Identity}: For \( f: A \to B \),
    \[
    f \circ \text{id}_A = f, \quad \text{id}_B \circ f = f.
    \]
\end{enumerate}
\end{defn}

\begin{thm}
The composition operation in a category \( \mathcal{C} \) is associative.
\end{thm}

\begin{proof}
By definition, for morphisms \( f: A \to B \), \( g: B \to C \), \( h: C \to D \) in \( \mathcal{C} \), the associativity axiom states \( (h \circ g) \circ f = h \circ (g \circ f) \). This is a direct consequence of the category’s structure, ensuring that the order of composition does not affect the result, as composition is defined to satisfy this property for all morphisms in \( \text{Hom}_{\mathcal{C}} \).
\end{proof}

Examples of categories include:
\begin{enumerate}
    \item \( \mathbf{Set} \): Objects are sets, morphisms are functions, composition is function composition, and identity morphisms are identity functions.
    \item \( \mathbf{Graph} \): Objects are graphs, morphisms are graph homomorphisms, composition is homomorphism composition, and identity morphisms preserve graph structure.
    \item \( \mathbf{Top} \): Objects are topological spaces, morphisms are continuous functions, with standard composition and identities.
\end{enumerate}

Functors map between categories, preserving their structure, and are crucial for relating the data fabric to its hypergraph representation.

\begin{defn}[Functor]
A \emph{functor} \( F: \mathcal{C} \to \mathcal{D} \) between categories \( \mathcal{C} \) and \( \mathcal{D} \) assigns:
\begin{enumerate}
    \item Each object \( A \in \text{Ob}(\mathcal{C}) \) to an object \( F(A) \in \text{Ob}(\mathcal{D}) \).
    \item Each morphism \( f: A \to B \) in \( \mathcal{C} \) to a morphism \( F(f): F(A) \to F(B) \) in \( \mathcal{D} \),
\end{enumerate}
such that:
\begin{enumerate}
    \item \( F(g \circ f) = F(g) \circ F(f) \) for \( f: A \to B \), \( g: B \to C \),
    \item \( F(\text{id}_A) = \text{id}_{F(A)} \) for each \( A \in \text{Ob}(\mathcal{C}) \).
\end{enumerate}
\end{defn}

\begin{thm}
A functor \( F: \mathcal{C} \to \mathcal{D} \) preserves composition and identities.
\end{thm}

\begin{proof}
By definition, \( F \) satisfies \( F(g \circ f) = F(g) \circ F(f) \) for morphisms \( f: A \to B \), \( g: B \to C \), ensuring composition is preserved. Similarly, \( F(\text{id}_A) = \text{id}_{F(A)} \) preserves identity morphisms. These properties are axioms of the functor, directly guaranteed by its definition.
\end{proof}

Natural transformations provide a way to compare functors, modeling relationships between different representations of the data fabric.

\begin{defn}[Natural Transformation]
A \emph{natural transformation} \( \eta: F \to G \) between functors \( F, G: \mathcal{C} \to \mathcal{D} \) assigns to each object \( A \in \text{Ob}(\mathcal{C}) \) a morphism \( \eta_A: F(A) \to G(A) \) in \( \mathcal{D} \), such that for every morphism \( f: A \to B \) in \( \mathcal{C} \), the following diagram commutes:
\[
\begin{tikzcd}
F(A) \arrow[r, "\eta_A"] \arrow[d, "F(f)"] & G(A) \arrow[d, "G(f)"] \\
F(B) \arrow[r, "\eta_B"] & G(B)
\end{tikzcd}
\]
i.e., \( \eta_B \circ F(f) = G(f) \circ \eta_A \).
\end{defn}

Category theory’s abstraction is particularly suited to data fabrics, where datasets (\( D \)) can be objects, transformations (\( T \)) can be morphisms, and operations like integration and navigation can be modeled as compositions or functorial mappings. Foundational references include \cite{MacLane1971} for a comprehensive treatment, \cite{Awodey2010} for an accessible introduction, and \cite{Spivak2014} for applications to databases and data management.


\begin{defn}[Data Fabric Category]
\label{subsec:df_category}
The \emph{data fabric category} \( \DF \) is defined as follows:

\begin{enumerate}
    \item \emph{Objects}: \( D = \{d_i(t)\}_{i,t} \), the time-indexed data assets (\cref{subsec:data_assets}).
  
    \item \emph{Morphisms}: \( T = \{t_1, \dots, t_m\} \), where \( t_i: d_i \to d_j \) is a transformation \( t_i: \Omega_i \to \Omega_j \) (\cref{subsec:transformations}).

    \item \emph{Composition}: For \( t_1: d_i \to d_j \), \( t_2: d_j \to d_k \), the composite \( t_2 \circ t_1: d_i \to d_k \), defined by \( (t_2 \circ t_1)(x) = t_2(t_1(x)) \).
  
    \item \emph{Identity}: For each \( d_i \in D \), the identity morphism \( \text{id}_{d_i}: d_i \to d_i \), where \( \text{id}_{d_i}(x) = x \).
\end{enumerate}
\end{defn}

\begin{thm}
\( \DF \) is a category, with associative composition and identity morphisms.
\end{thm}

\begin{proof}
To verify \( \DF \) is a category:
\begin{enumerate}
    \item \emph{Associativity}: For morphisms \( t_1: d_i \to d_j \), \( t_2: d_j \to d_k \), \( t_3: d_k \to d_l \), composition is function composition: \( (t_3 \circ t_2) \circ t_1(x) = t_3(t_2(t_1(x))) = t_3 \circ (t_2 \circ t_1)(x) \), which is associative by the associativity of function composition.
    \item \emph{Identity}: For \( t_i: d_i \to d_j \), the identity \( \text{id}_{d_i}(x) = x \) satisfies \( t_i \circ \text{id}_{d_i} = t_i \), as \( t_i(\text{id}_{d_i}(x)) = t_i(x) \), and similarly \( \text{id}_{d_j} \circ t_i = t_i \). For \( \text{id}_{d_j}: d_j \to d_j \), \( \text{id}_{d_j}(t_i(x)) = t_i(x) \).
\end{enumerate}
Thus, \( \DF \) satisfies the category axioms.
\end{proof}

For example, in a supply chain fabric, \( d_i \in D \) is raw shipment data, \( d_j \) is normalized data, and \( d_k \) is an aggregated inventory summary. Transformations \( t_1: d_i \to d_j \) (normalization) and \( t_2: d_j \to d_k \) (aggregation) compose as \( t_2 \circ t_1: d_i \to d_k \), representing the full data processing pipeline.

The hypergraph \( G \) (\cref{subsec:hypergraph_prelim}) is modeled via a functor to the category of hypergraphs.

\begin{defn}[Hypergraph Category]
The category \( \mathcal{HG} \) has:
\begin{enumerate}
    \item \emph{Objects}: Hypergraphs \( G = (V, E) \), where \( V \) is a set of vertices, and \( E \subseteq \mathcal{P}(V) \) is a set of hyperedges.
    \item \emph{Morphisms}: Hypergraph homomorphisms \( \phi: G_1 \to G_2 \), where \( \phi: V_1 \to V_2 \) maps vertices such that for each \( e \in E_1 \), \( \phi(e) = \{\phi(v) \mid v \in e\} \in E_2 \).
    \item \emph{Composition}: Standard function composition of homomorphisms.
    \item \emph{Identity}: The identity morphism \( \text{id}_G: G \to G \), where \( \text{id}_G(v) = v \).
\end{enumerate}
\end{defn}

A functor \( F: \DF \to \mathcal{HG} \) maps the data fabric to its hypergraph representation:
\begin{enumerate}
    \item \emph{Objects}: \( d_i \in D \mapsto v_i \in V \), where \( v_i \) is a vertex in \( G \).
    \item \emph{Morphisms}: \( t_i: d_i \to d_j \mapsto e \in E \), where \( e = (T_e, H_e) \) with \( T_e = \{v_i, m_k\} \), \( H_e = \{v_j\} \), and \( m_k \in M \) is metadata associated with \( t_i \).
\end{enumerate}

The functor preserves composition:
\[
F(t_2 \circ t_1) = F(t_2) \circ F(t_1),
\]
as a sequence \( t_2 \circ t_1: d_i \to d_k \) maps to a hyperedge path from \( v_i \) to \( v_k \) via \( v_j \). For instance, \( t_1 \circ t_2 \) in a supply chain fabric corresponds to a path in \( G \) linking shipment data, inventory, and forecasts, enabling lineage tracking.

\begin{thm}
The functor \( F: \DF \to \mathcal{HG} \) preserves composition and identities.
\end{thm}

\begin{proof}
For morphisms \( t_1: d_i \to d_j \), \( t_2: d_j \to d_k \), \( F(t_1) = e_1 = (\{v_i, m_{k1}\}, \{v_j\}) \), \( F(t_2) = e_2 = (\{v_j, m_{k2}\}, \{v_k\}) \). The composite \( t_2 \circ t_1 \) maps to a hyperedge path \( e_1, e_2 \), where \( H_{e_1} \cap T_{e_2} = \{v_j\} \neq \emptyset \), represented as a composite morphism in \( \mathcal{HG} \). Thus, \( F(t_2 \circ t_1) = F(t_2) \circ F(t_1) \). For identity \( \text{id}_{d_i} \), \( F(\text{id}_{d_i}) = \text{id}_{v_i} \), the identity morphism in \( \mathcal{HG} \), preserving identities.
\end{proof}


The categorical structure of \( \DF \) and the functor \( F \) unify the operations of the data fabric (\cref{sec-3}):

\begin{enumerate}
    \item \emph{Data Integration}: Composes morphisms in \( \DF \). For \( t_1: d_i \to d_j \), integration chains transformations to align schemas, minimizing schema distance (\cref{data_integration}).
    \item \emph{Metadata-Driven Navigation}: Traverses hyperedge paths in \( \mathcal{HG} \) via \( F \), finding shortest paths between datasets and metadata (\cref{navigation}).
    \item \emph{Provenance Tracking}: Reconstructs morphism sequences in \( \DF \), yielding \( \tr (d_i) \) (\cref{provenance}).
    \item \emph{Federated Learning}: Models local analytics as morphisms in \( \DF \), with aggregation as a functorial operation across nodes (\cref{subsec:federated_learning}).
\end{enumerate}

Consider a healthcare fabric: patient records \( d_i \) are transformed to normalized data \( d_j \) via \( t_1 \) (cleaning), then to predictions \( d_k \) via \( t_2 \) (model inference). The sequence \( t_2 \circ t_1 \) in \( \DF \) maps to a hyperedge path in \( G \), supporting integration (aligning records), navigation (locating related data), provenance (tracking transformations), and federated learning (distributed model training).


\textbf{Natural transformations} model relationships between different functorial representations of the data fabric, addressing consistency constraints in distributed systems (\cref{sec-6}).

\begin{defn}[Data Fabric Natural Transformation]
\label{subsec:natural_transformations}
A \emph{natural transformation} \( \eta: F \to G \) between functors \( F, G: \DF \to \mathcal{HG} \) assigns to each dataset \( d_i \in D \) a morphism \( \eta_{d_i}: F(d_i) \to G(d_i) \) in \( \mathcal{HG} \), such that for each transformation \( t_i: d_i \to d_j \):
\[
\eta_{d_j} \circ F(t_i) = G(t_i) \circ \eta_{d_i}.
\]
\end{defn}

For example, \( F \) and \( G \) may map \( \DF \) to different hypergraph representations of \( G \) (e.g., with different metadata granularity). The natural transformation \( \eta \) ensures that navigation paths (via \( F \)) align with transformation sequences (via \( G \)), maintaining consistency across distributed nodes \( N \in \Sigma \). This is critical for operations like provenance tracking, where transformation histories must be consistent across representations.

\begin{thm}
Natural transformations \( \eta: F \to G \) ensure commutative diagrams for data fabric operations.
\end{thm}

\begin{proof}
By definition, for \( t_i: d_i \to d_j \), \( \eta_{d_j} \circ F(t_i) = G(t_i) \circ \eta_{d_i} \). This ensures that the diagram in \( \mathcal{HG} \) commutes, meaning that transformations in \( \DF \) (mapped by \( F \) and \( G \)) preserve the relational structure of \( G \). For operations like navigation, \( F(t_i) \) and \( G(t_i) \) represent hyperedge paths, and \( \eta \) ensures path equivalence, maintaining operational consistency.
\end{proof}


The categorical perspective provides several benefits for data fabrics:
\begin{enumerate}
    \item \emph{Unification}: Operations are abstracted as compositions or functorial mappings, simplifying their analysis and implementation.
    \item \emph{Consistency}: Natural transformations model consistency constraints, ensuring alignment across distributed nodes (\cref{sec-6}).
    \item \emph{Scalability}: Functorial mappings to \( \mathcal{HG} \) support efficient navigation and provenance tracking, though dynamic updates to \( \DF \) (e.g., evolving schemas \( S_i(t) \)) pose challenges (\cref{subsec:heterogeneity}).
\end{enumerate}

For instance, in an IoT fabric, sensor data transformations are morphisms in \( \DF \), mapped to hyperedges in \( G \), enabling real-time analytics (\cref{subsec:realtime}). However, challenges like NP-hard schema matching (\cref{subsec:heterogeneity}) and latency constraints (\cref{subsec:realtime}) require categorical extensions, such as dynamic categories or adjoint functors, as explored in \cite{Spivak2014}.

The categorical approach also informs distributed system design by abstracting node interactions in \( \Sigma \). Future work may leverage monoidal categories or topos theory to model governance policies \( P \) or analytical functions \( A \), building on frameworks proposed in \cite{Schultz2016}.

\section{Hypergraph Embeddings and Representations} \label{sec-5}

The hypergraph \( G = (V, E) \), a core component of the data fabric 
\[
 \cF = (D, M, G, T, P, A),
 \]
  encodes multi-way relationships among data assets \( D \) and metadata \( M \). This section offers a rigorous mathematical treatment of three key features: vector embeddings of data assets, the adjacency matrix construction, and the embedding of \( G \) into a modular tensor category (MTC). These build on the operations in \cref{sec-3} (e.g., integration and navigation) and the categorical framework in \cref{sec-4} (datasets as objects, transformations as morphisms). 


\subsection{Vector Representations of Data Assets} \label{subsec:vector_representations}
To integrate data assets \( D = \{d_i(t)\}_{i,t} \) into the hypergraph \( G \) and its categorical embedding, we represent each dataset \( d_i(t) \), defined over domain \( \Omega_i \) with schema \( S_i \), as a vector in a finite-dimensional space. This accommodates numerical, categorical, and mixed-type data while capturing temporal dynamics, supporting operations like integration and navigation.

For numerical data, where \( \Omega_i \subseteq \R^k \), a dataset \( d_i(t) \) consists of points \( \{\x_1, \dots, \x_n\} \), each \( \x_j \in \R^k \). The representative vector is the centroid:
\[
\bv_i(t) = \frac{1}{n} \sum_{j=1}^n \x_j \in \R^k,
\]
assuming \( \Omega_i \) is convex. For non-convex domains, the medoid is used:
\[
\bv_i(t) = \arg\min_{\x_j \in d_i(t)} \sum_{l=1}^n \|\x_j - \x_l\|_2,
\]
computed in \( O(n^2 k) \) time. This preserves central tendency for similarity measures.

For categorical data, where \( \Omega_i = \{c_1, \dots, c_m\} \), we embed via \( \phi: \Omega_i \to \R^d \) (e.g., one-hot encoding \( \phi(c_j) = \e_j \in \R^m \) or learned embeddings with \( d \ll m \), normalized to \( \|\phi(c_j)\|_2 = 1 \)). The representative vector is the frequency-weighted average:
\[
\bv_i(t) = \sum_{j=1}^m f_j \phi(c_j) \in \R^d,
\]
where \( f_j = |\{ \x \in d_i(t) \mid \x = c_j \}| / n \).

For mixed-type data, combining numerical \( \x_{\text{num}} \in \R^k \) and categorical attributes, the vector is:
\[
\bv_i(t) = \left( \x_{\text{num}}, \phi(c_{j_1}), \dots, \phi(c_{j_l}) \right) \in \R^{k + l d},
\]
with total dimension standardized to \( \leq 15 \) for operational efficiency.

Temporal dynamics are modeled by \( \bv_i: T \to \R^p \) (where \( p = k \), \( d \), or \( k + l d \)). For streaming updates:
\[
\bv_i(t + \Delta t) = (1 - \alpha) \bv_i(t) + \alpha \x_{\text{new}},
\]
with \( \alpha \in (0,1) \) and \( \x_{\text{new}} \) appropriately embedded, in \( O(p) \) time.

Metadata \( m_j = (d_i, \alpha_j, \tau_j) \in M \) is vectorized as:
\[
\bv_{m_j} = (\phi(\alpha_j), \bv_{\tau_j}) \in \R^{q + s},
\]
where \( \phi(\alpha_j) \in \R^q \) embeds attributes and \( \bv_{\tau_j} \in \R^s \) encodes history (e.g., average transformation parameters), ensuring sparsity and bounded dimension.

\subsection{Construction of the Adjacency Structure} \label{subsec:adjacency_structure}

The adjacency structure of the hypergraph \( G = (V, E) \) defines connectivity through directed hyperedges, enabling multi-way relationships critical for operations like data integration, navigation, and provenance tracking. We formalize how vertices \( V = D \cup M \) are linked via hyperedges \( E \), grounding the structure in the categorical framework of \cref{sec-4}.

A hyperedge \( e \in E \) is a directed pair \( e = (T_e, H_e) \), where \( T_e \subseteq V \) is the tail (input vertices) and \( H_e \subseteq V \) is the head (output vertices). The adjacency set for a vertex \( v \in V \) is:
\[
\Adj(v) = \{ e \in E \mid v \in T_e \cup H_e \},
\]
with incoming and outgoing components:
\[
\text{In}(v) = \{ e \in E \mid v \in H_e \}, \quad \text{Out}(v) = \{ e \in E \mid v \in T_e \}.
\]
Incidence matrices are defined as:
\[
I_T \in \{0,1\}^{|V| \times |E|}, \quad (I_T)_{v,e} = 1 \text{ if } v \in T_e, \quad 0 \text{ otherwise},
\]
\[
I_H \in \{0,1\}^{|V| \times |E|}, \quad (I_H)_{v,e} = 1 \text{ if } v \in H_e, \quad 0 \text{ otherwise}.
\]
For large \( G \), sparsity is enforced, with non-zero entries in \( I_T \) and \( I_H \) bounded by \( O(|V| \log |V|) \), enabling efficient storage in compressed sparse row (CSR) formats with access complexity \( O(1) \) per entry.

Hyperedges are generated based on operational dependencies, aligning with the morphisms in \( \DF \):  

\begin{enumerate}
    \item \emph{Integration}: A transformation \( t_i \in T \) mapping \( \{d_{i_1}, \dots, d_{i_n}\} \to d_j \) induces \( e = (\{v_{i_1}, \dots, v_{i_n}, v_{m_k}\}, \{v_j\}) \), where \( m_k \in M \) records \( t_i \) in \( \tau_k(d_j) \).
    \item \emph{Navigation}: Datasets \( d_i, d_j \) sharing attributes in \( \alpha_k \) form \( e = (\{v_i, v_j\}, \{v_{m_k}\}) \).
    \item \emph{Provenance}: A derived dataset \( d_j \) linked to sources \( \{d_{i_1}, \dots, d_{i_n}\} \) yields \( e = (\{v_{i_1}, \dots, v_{i_n}, v_{m_j}\}, \{v_j\}) \).
    \item \emph{Federated Learning}: Local analytics on \( D_n = \{d_{i_1}, \dots, d_{i_k}\} \) producing \( \theta_n \) create \( e = (\{v_{i_1}, \dots, v_{i_k}, v_{m_n}\}, \{v_{\theta_n}\}) \).
\end{enumerate}

\begin{figure}[h!]
    \centering
    \begin{tikzpicture}[node distance=2cm, auto]
        \node[draw, circle] (v1) {$v_i$};
        \node[draw, circle, above right=of v1] (vm) {$v_m$};
        \node[draw, circle, below right=of vm] (v2) {$v_j$};
        \draw[->, thick] (v1) to[bend left=30] node[midway, above] {$e$} (v2);
        \draw[->, thick] (vm) to[bend right=30] (v2);
        \node[below=0.5cm of v2] {\small Directed hyperedge \( e = (\{v_i, v_m\}, \{v_j\}) \in E \)};
    \end{tikzpicture}
    \caption{Structure of the hypergraph \( G \), with vertices \( v_i, v_m \in V \) (datasets or metadata) and hyperedge \( e \in E \) encoding a multi-way relationship.}
    \label{fig:hypergraph_structure}
\end{figure}

Sparsity is ensured by limiting each vertex to \( O(\log |V|) \) hyperedges, yielding \( |E| = O(|V| \log |V|) \). This is achieved by prioritizing relationships with mutual information:
\[
I(t_i(d_i); d_j) \geq \epsilon,
\]
where \( \epsilon > 0 \), approximated by cosine similarity \( \cos(\bv_i(t), \bv_j(t)) \in [-1, 1] \), computed in \( O(p) \). Connectivity is maintained by ensuring each dataset vertex \( v_i \) is reachable from at least one metadata vertex \( v_{m_j} \), supporting navigation queries.
Paths in \( G \) are sequences \( (e_1, e_2, \dots, e_m) \) with \( H_{e_i} \cap T_{e_{i+1}} \neq \emptyset \), enumerated via a hypergraph-adapted Dijkstra’s algorithm, with complexity \( O(|E| + |V| \log |V|) \), aligning with navigation efficiency.

\subsection{Mapping to Modular Tensor Categories} \label{subsec:mtc_mapping}

To capture the hypergraph’s multi-way dependencies algebraically, we embed \( G = (V, E) \) into a modular tensor category (MTC), a semisimple, spherical ribbon fusion category with braided monoidal structure and non-degenerate modularity. This embedding extends the categorical framework of \cref{sec-4}, where datasets and transformations are objects and morphisms in \( \DF \), providing a precise algebraic model for operations.

\begin{defn} \label{defn:mtc}
A \emph{modular tensor category} \( \mathcal{C} \) is a semisimple, spherical ribbon fusion category equipped with:
\begin{enumerate}
    \item A finite set of \emph{simple objects} \( \{X_i\} \), closed under a tensor product \( \otimes: \mathcal{C} \times \mathcal{C} \to \mathcal{C} \), with unit object \( \mathbf{1} \).
    \item \emph{Morphisms} \( \text{Hom}_{\mathcal{C}}(X, Y) \), with composition and identity satisfying category axioms.
    \item A \emph{braiding}, a natural isomorphism \( c_{X,Y}: X \otimes Y \to Y \otimes X \), satisfying the hexagon axioms:
    \[
    c_{X, Y \otimes Z} = (1_Y \otimes c_{X,Z}) \circ (c_{X,Y} \otimes 1_Z),
    \]
    \[
    c_{X \otimes Y, Z} = (c_{X,Z} \otimes 1_Y) \circ (1_X \otimes c_{Y,Z}).
    \]
    \item A \emph{ribbon twist}, a natural isomorphism \( \theta_X: X \to X \), compatible with braiding:
    \[
    \theta_{X \otimes Y} = c_{Y,X} \circ c_{X,Y} \circ (\theta_X \otimes \theta_Y).
    \]
    \item A non-degenerate \emph{\( S \)-matrix}, defined by:
    \[
    S_{X,Y} = \text{tr}(c_{Y,X} \circ c_{X,Y}),
    \]
    where \( \text{tr} \) is the quantum trace, and \( S \) is invertible over \( \C \).
\end{enumerate}
\end{defn}

The embedding maps vertices \( v_i \in V \), representing datasets \( d_i(t) \in D \) or metadata \( m_j \in M \), to simple objects \( X_i \in \mathcal{C} \). The vector representation \( \bv_i(t) \in \R^p \) or \( \bv_{m_j} \in \R^q \) informs \( X_i \), e.g., as the vector space \( \C^p \) equipped with a tensor product structure. The tensor product models joint dependencies:
\[
X_i \otimes X_j \cong \bigoplus_k N_{ij}^k X_k,
\]
where fusion coefficients \( N_{ij}^k = 1 \) if \( v_k \in H_e \) for some \( e \in E \) with \( v_i, v_j \in T_e \), else 0, reflecting hyperedge connectivity.

Each hyperedge \( e = (T_e, H_e) \), with \( T_e = \{v_{i_1}, \dots, v_{i_n}\} \), \( H_e = \{v_{j_1}, \dots, v_{j_m}\} \), corresponds to a morphism:
\[
f_e: X_{i_1} \otimes \cdots \otimes X_{i_n} \to X_{j_1} \otimes \cdots \otimes X_{j_m}.
\]
For data integration (\cref{data_integration}), where \( H_e = \{v_j\} \), \( f_e: \otimes_{k=1}^n X_{i_k} \to X_j \) encodes the transformation matrix; for navigation (\cref{navigation}), \( f_e \) projects onto metadata objects.

The ribbon twist captures cyclic dependencies:
\[
\theta_{X_i} = \sum_{e \in \text{Loop}(v_i)} w_e f_e,
\]
where \( \text{Loop}(v_i) \subseteq E \) is the set of hyperedge cycles containing \( v_i \), and \( w_e = I(t_e(d_i); d_j) \) for transformations \( t_e \), or frequency otherwise, satisfying:
\[
\theta_{X_i \otimes X_j} = c_{Y,X} \circ c_{X,Y} \circ (\theta_X \otimes \theta_Y).
\]
The \( S \)-matrix quantifies connectivity:
\[
S_{X_i, X_j} = \sum_{e \in E_{i,j}} \text{tr}(c_{X_j, X_i} \circ c_{X_i, X_j}),
\]
with \( E_{i,j} = \{ e \in E \mid v_i, v_j \in T_e \cup H_e \} \). For large hypergraphs, \( S \) is approximated via the Laplacian:
\[
L = I_T I_T^T - I_H I_H^T,
\]
with eigenvalues reflecting connectivity, computed in \( O(|V| \log |V|) \) using randomized SVD.

\begin{thm} \label{thm:mtc_embedding}
The hypergraph \( G \) embeds into an MTC \( \mathcal{C} \) via a faithful functor \( \Phi: \mathcal{HG} \to \mathcal{C} \), preserving vertices as simple objects, hyperedges as morphisms, and adjacency through braiding and the \( S \)-matrix.
\end{thm}

\begin{proof}
Define \( \Phi: \mathcal{HG} \to \mathcal{C} \):

\begin{enumerate}
    \item \emph{Objects}: Map \( G = (V, E) \) to \( \mathcal{C}_G \subseteq \mathcal{C} \), with simple objects \( \{X_i \mid v_i \in V\} \).
    \item \emph{Morphisms}: Map a hypergraph homomorphism \( h: G_1 \to G_2 \) to \( \Phi(h): \mathcal{C}_{G_1} \to \mathcal{C}_{G_2} \), sending \( X_i \to X_{h(i)} \), \( f_e \to f_{h(e)} \).
\end{enumerate}

For a hyperedge \( e = (T_e, H_e) \), \( \Phi(e) = f_e: \otimes_{v_k \in T_e} X_k \to \otimes_{v_l \in H_e} X_l \). A path \( (e_1, e_2) \) with \( H_{e_1} \cap T_{e_2} \neq \emptyset \) maps to \( f_{e_2} \circ f_{e_1} \), preserving composition. The braiding \( c_{X_i, X_j} \) satisfies hexagon axioms, modeling hyperedge symmetries critical for navigation (\cref{navigation}). The twist \( \theta_{X_i} \) encodes cycles, supporting provenance (\cref{provenance}). The \( S \)-matrix, with non-degeneracy ensured by \( G \)’s connectivity, quantifies vertex interactions. Since \( \Phi \) preserves objects, morphisms, and their relational structure, it is faithful, embedding \( G \) into \( \mathcal{C} \).
\end{proof}

The MTC’s braiding \( c_{X_i, X_j} \) derives from the R-matrix of a quantum group, such as \( U_q(\mathfrak{sl}_2) \) at a root of unity, which encodes non-commutative symmetries \cite{Turaev1994}. This quantum group structure underpins the \( S \)-matrix and fusion rules, suggesting that quantum-inspired algorithms could optimize operations like navigation and provenance tracking by exploiting these algebraic symmetries, as further explored in \cref{subsec:scalability,subsec:realtime}.


\subsection{Braiding Action and Geometric Analogies} \label{subsec:braiding_action}
The braiding in the MTC \( \mathcal{C} \), a natural isomorphism \( c_{X,Y}: X \otimes Y \to Y \otimes X \) satisfying the hexagon axioms, models symmetries in hyperedges of \( G \). For hyperedge \( e = (T_e, H_e) \) with \( T_e = \{v_i, v_j, \dots\} \), braiding ensures input order invariance, preserving morphism \( f_e \). This is crucial for metadata-driven navigation, where datasets linked by metadata form \( e = (\{v_i, v_j\}, \{v_{m_k}\}) \), and swapping inputs maintains the relation:
\[
f_e \circ c_{X_i, X_j} = f_e.
\]
Such invariance supports permutation-invariant queries and transformation sequences in provenance tracking.

To rigorize the geometric analogy, recall that MTCs give rise to modular functors in topological quantum field theories (TQFTs), assigning finite-dimensional vector spaces to labeled surfaces and linear maps to cobordisms, satisfying gluing axioms and yielding projective representations of mapping class groups (MCGs) on homology groups of configuration spaces. In 3D TQFTs like Reshetikhin-Turaev, modular functors encode conformal blocks on genus-g surfaces with n punctures, with braiding from half-monodromies around punctures. In Hurwitz spaces \( \Hur_{d,n} \), parametrizing degree-d covers of \( \mathbb{P}^1 \) with n branch points and fixed ramification profile, the braid group \( B_n = \pi_1(\Conf_n(\mathbb{C})/S_n) \) acts transitively on connected components for simple ramification (Hurwitz theorem). For cover \( \phi: C \to \mathbb{P}^1 \) with branch points \( p = (p_1, \dots, p_n) \in \Conf_n(\mathbb{C}) \), monodromy \( \rho_p: \pi_1(\mathbb{P}^1 \setminus p, b) \to S_d \) is conjugated under \( B_n \):
\[
\rho'(\gamma) = \beta \rho(\gamma) \beta^{-1}, \quad \beta \in S_d,
\]
where conjugation arises from paths in \( \Conf_n(\mathbb{C}) \). This parallels MTC braiding, where \( c_{X_i, X_j} \) induces isomorphisms preserving compositions, akin to modular functor maps under MCG actions.

Formally, braiding satisfies the Yang-Baxter equation:
\[
(c_{X,Y} \otimes \id_Z) (\id_X \otimes c_{Y,Z}) (c_{X,Z} \otimes \id_Y) = (\id_Y \otimes c_{X,Z}) (c_{X,Y} \otimes \id_Z) (\id_X \otimes c_{Y,Z}),
\]
corresponding to braid relation. For multi-vertex edges \( e = (\{v_{i_1}, v_{i_2}, v_{i_3}\}, \{v_j\}) \):
\[
f_e \circ (c_{X_{i_1}, X_{i_2}} \otimes \id_{X_{i_3}}) = f_e,
\]
by naturality and coherence.

In moduli spaces \( \mathcal{M}_{g,n} \), the MCG \( \Mod_{g,n} \) acts on Teichmüller space via Dehn twists; for $g=0$, this reduces to $B_{n+1}$. Modular functors from MTCs represent \( \Mod_{g,n} \), with Dehn twists as exponentials of Casimir operators in quantum group reps. In the fabric, braiding models dependency "twists," preserving outputs via coherence.

Algebraically, MTCs are semisimple ribbon categories with non-degenerate $S$-matrix; their modular functors extend $B_n$ reps to $SL(2,Z)$ via Kirby coloring in TQFTs. Unitary MTCs from quantum groups at level $k$ yield finite-dimensional reps, preserving norms for federated learning.

Knot theory: Closed braids yield links; invariants like Jones $V_L(q) = \tr (\rho(\beta))$ where $\beta$ braid, detect anomalies.

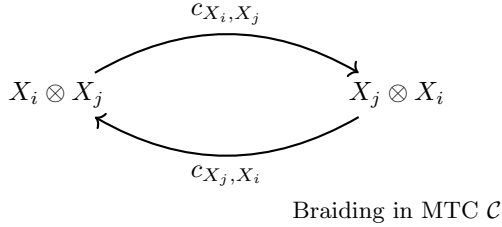
\begin{figure}[h!]
    \centering
    \begin{tikzpicture}[node distance=2cm, auto]
        \node (XiXj) {$X_i \otimes X_j$};
        \node[right=3cm of XiXj] (XjXi) {$X_j \otimes X_i$};
        \draw[->, thick] (XiXj) to[bend left=30] node[midway, above] {$c_{X_i, X_j}$} (XjXi);
        \draw[->, thick] (XjXi) to[bend left=30] node[midway, below] {$c_{X_j, X_i}$} (XiXj);
        \node[below=1cm of XjXi] {\small Braiding in MTC \( \mathcal{C} \)};
    \end{tikzpicture}
    \caption{Braiding modeling symmetries.}
    \label{fig:braiding_action}
\end{figure}

\begin{lem} \label{lem:braiding_invariance}
Braiding \( c_{X_i, X_j} \) preserves hyperedge semantics.
\end{lem}

\begin{proof}
Let \( f_e: \bigotimes_{k \in T_e} X_k \to \bigotimes_{l \in H_e} X_l \). By naturality of c, the diagram commutes:
\[
\begin{CD}
\bigotimes T_e @>f_e>> \bigotimes H_e \\
@Vc_{X_i, X_j} \otimes \id VV @V\id VV \\
\bigotimes T_e @>f_e>> \bigotimes H_e.
\end{CD}
\]
Thus \( f_e \circ (c_{X_i, X_j} \otimes \id) = f_e \). For arbitrary permutations, coherence theorems ensure the braiding generates consistent isomorphisms (Mac Lane's coherence theorem for monoidal categories). For navigation, equivalent paths under braid group actions; for provenance, sequences invariant as monoidal structure preserves compositions via functoriality.
\end{proof}

\subsection{Embedding and Applications} \label{subsec:mtc_embedding}

Motivated by the need to capture relational symmetries in a unified algebraic structure, particularly for future quantum data applications where braided monoidal symmetries can model entanglement and topological invariants, we embed the hypergraph \( G \) into a modular tensor category (MTC) \( \mathcal{C} \). This embedding extends the categorical framework of \cref{sec-4}, providing tools for symmetry-based optimizations in NP-hard tasks (\cref{sec-6}) and fault-tolerant operations under CAP/CAL (\cref{sec-7}). By leveraging MTCs, the framework gains access to braided structures for permutation-invariant processing, knot invariants for anomaly detection, and modular functors for potential quantum-scale computations, aligning with trends in quantum data centers and post-quantum security.

The embedding is formalized via a faithful functor \( \Phi: \mathcal{HG} \to \mathcal{C} \), where \( \mathcal{HG} \) is the category of directed hypergraphs with objects as hypergraphs \( G = (V, E) \) and morphisms as hypergraph homomorphisms preserving tails and heads (\cref{subsec:df_category}):
\begin{enumerate}
    \item \emph{Objects}: \( \Phi(G) = \bigoplus_{v_i \in V} X_i \), where each \( X_i \) is a simple object in \( \mathcal{C} \) corresponding to vertex \( v_i \), forming a subcategory \( \mathcal{C}_G \subseteq \mathcal{C} \).
    \item \emph{Morphisms}: A hypergraph homomorphism \( h: G_1 \to G_2 \) maps to \( \Phi(h): \mathcal{C}_{G_1} \to \mathcal{C}_{G_2} \), sending \( X_i \to X_{h(i)} \) and \( f_e \to f_{h(e)} \), where \( f_e \) is the morphism associated with hyperedge e.
\end{enumerate}

\begin{prop} \label{prop:faithful_functor}
The functor \( \Phi \) is faithful, preserving injections of morphisms.
\end{prop}

\begin{proof}
Since \( \mathcal{C} \) is semisimple and ribbon, simple objects \( X_i \) have distinct quantum dimensions, and morphisms \( f_e \) are determined by their matrix elements in the basis of simple paths. Hypergraph homomorphisms h induce unique labelings on objects and edges, ensuring \( \Phi \) injects $\Hom_{\mathcal{HG}}(G_1, G_2)$ into $\Hom_{\mathcal{C}}(\mathcal{C}_{G_1}, \mathcal{C}_{G_2})$, as non-isomorphic hypergraphs yield distinct fusion rules.
\end{proof}

Applications of this embedding include:

\begin{enumerate}
    \item \emph{Navigation}: Morphism compositions \( f_{e_2} \circ f_{e_1} \) represent hyperedge paths, aligning with shortest-path queries. In MTC, braiding ensures path equivalence up to isomorphism, enabling efficient traversal via modular invariants, with complexity reduced to \( O(|V| \log |V|) \) using S-matrix diagonalization.
 
    \item \emph{Provenance}: The ribbon twist \( \theta_{X_i} \in \Aut(X_i) \), satisfying \( \theta_{X \otimes Y} = c_{Y,X} \circ c_{X,Y} \circ (\theta_X \otimes \theta_Y) \), models cyclic dependencies. For a cycle \( d_i \to d_j \to d_i \), the trace $\tr(\theta_{X_i \otimes X_j})$ yields invariants detecting non-trivial loops, supporting transformation sequences.

    \item \emph{Federated Learning}: Morphisms \( f_e \) aggregate local analytics while preserving privacy through pivotal structures in MTCs, where duals \( X^* \) enable adjoint operations for gradient sharing. This ensures coherence in distributed updates, with quantum extensions potentially leveraging anyonic fusion for secure multi-party computation.
\end{enumerate}

The S-matrix \( S_{ij} = \frac{1}{\mathcal{D}} \tr(c_{X_j, X_i^*} \circ c_{X_i, X_j}) \), where \( \mathcal{D} \) is the total quantum dimension, is approximated via randomized SVD on the hypergraph Laplacian \( L = I - D^{-1/2} A D^{-1/2} \) (A from incidence), achieving \( O(|V| \log |V|) \) time with error \( \epsilon \) via rank-k approximation (k << |V| for sparse G). This enhances scalability for large fabrics (\cref{subsec:scalability}).


\section{Computational Challenges} \label{sec-6}
The data fabric framework \( \cF = (D, M, G, T, P, A) \), operating over the distributed system \( \Sigma = (N, C) \) (\cref{sec-2}), enables sophisticated operations such as data integration, metadata-driven navigation, and federated learning (\cref{sec-3}). However, its mathematical complexity, rooted in the hypergraph \( G \) (\cref{subsec:adjacency_structure}) and its embedding into a modular tensor category (MTC) (\cref{subsec:mtc_mapping}), introduces significant computational challenges. Heterogeneity in data assets, scalability across distributed nodes, governance through policy enforcement, and real-time processing demands impede interoperability, efficiency, compliance, and responsiveness. This section analyzes these challenges, formulating them within the categorical structure \( \DF \) (\cref{subsec:df_category}) and the MTC’s braided monoidal framework (\cref{subsec:braiding_action}), providing rigorous complexity bounds and NP-hardness proofs. We propose mitigation strategies leveraging the hypergraph’s adjacency properties and MTC symmetries, 
%
and draw connections to topological and spectral methods to address these bottlenecks, ensuring alignment with the operational and categorical foundations.

\subsection{Heterogeneity} \label{subsec:heterogeneity}
Heterogeneity in data assets \( D = \{d_i(t)\}_{i,t} \), characterized by diverse schemas \( S_i \) and domains \( \Omega_i \), complicates data integration and metadata-driven navigation (\cref{data_integration} and \cref{navigation}). Aligning disparate representations requires transformations \( t_i \in T \), modeled as morphisms in the data fabric category \( \DF \) (\cref{subsec:df_category}). The hypergraph \( G \), with vertices \( V = D \cup M \) and hyperedges \( E \), encodes these relationships using vector representations (\cref{subsec:vector_representations}), but stochastic distributions and dynamic schema evolution pose computational hurdles.

\begin{defn}[Heterogeneous Data Integration]
Heterogeneous data integration unifies datasets \( d_i, d_j \in D \) with distinct schemas \( S_i \neq S_j \) and domains \( \Omega_i, \Omega_j \) via transformations \( t_i \in T \), minimizing expected loss under stochastic conditions while ensuring semantic compatibility.
\end{defn}

The optimization problem seeks a transformation \( t_i \in T \) minimizing:
\[
\min_{t_i \in T} \mathbb{E}[\mathcal{L}(t_i(d_i), d_j) \mid P(d_i, d_j)],
\]
where \( \mathcal{L} = W_2(P_{t_i(d_i)}, P_{d_j}) \) is the 2-Wasserstein distance between distributions. For \( n \)-point distributions in \( \R^k \), computing \( W_2 \) involves solving an optimal transport problem, with complexity:
$
O(n^3 \log n),
$
derived from linear programming over \( n \times n \) transport matrices, straining scalability for high-dimensional \( \Omega_i \). In the MTC embedding (\cref{subsec:mtc_mapping}), \( t_i \) corresponds to a morphism \( f_e: X_i \to X_j \), and the loss relates to fusion coefficients \( N_{ij}^k \), computed iteratively for dynamic schemas with complexity \( O(|E|) \).

Dynamic schema matching, essential for evolving \( S_i(t) \), maximizes attribute similarity:
\[
\max_{\pi: S_i \to S_j} \mathbb{E}[\text{sim}(a, \pi(a)) \mid \text{ont}(S_i)],
\]
where \( \text{ont}(S_i) \) is a noisy ontology, and \( \text{sim}(a, b) \in [0, 1] \). This is NP-hard, as shown in \cref{data_integration}, reducing subgraph isomorphism to schema matching. The categorical perspective models \( S_i(t) \) as dynamic objects in \( \DF \), but adapting morphisms requires \( O(|S_i|) \) per update, a bottleneck for large schemas. For unstructured data, embeddings via optimal transport (\cref{subsec:vector_representations}) scale as \( O(n \log n) \), impacting federated learning (\cref{subsec:federated_learning}). Sinkhorn’s algorithm reduces this to \( O(n^2) \), with accuracy trade-offs bounded by:
\[
I(t_i(d_i); d_j) \leq I(d_i; d_j) - \mathcal{L}(t_i),
\]
where \( I \) is mutual information. The MTC’s braiding action (\cref{subsec:braiding_action}) suggests permutation-invariant morphisms to simplify schema alignment, potentially leveraging topological invariants to reduce complexity, as explored in \cref{sec:conclusion}.

The MTC’s fusion coefficients, derived from quantum group representations, offer a quantum-inspired approach to schema matching. By modeling datasets as representations of a quantum group like \( U_q(\mathfrak{sl}_2) \), transformations \( t_i \) act as intertwiners, potentially reducing the complexity of aligning noisy ontologies by exploiting non-commutative symmetries, as further discussed in \cref{sec:conclusion}.

\subsection{Scalability} \label{subsec:scalability}
Scalability is critical for processing large datasets \( D \) across \( \Sigma \), as required for cloud-based analytics platforms. The hypergraph \( G \), with sparse incidence matrices \( I_T, I_H \) (\cref{subsec:adjacency_structure}), supports efficient navigation, but dynamic workloads and node variability challenge distribution efficiency (\cref{scalability_distribution}).

\begin{defn}[Scalable Data Distribution]
Scalable data distribution partitions \( D = \bigcup_{n \in N} D_n \), where \( D_n \subseteq D \) resides on node \( n \in N \), to minimize computational cost \( \sum_{n \in N} \text{cost}(a(D_n)) \) and communication cost \( \sum_{(n_i, n_j) \in C} \text{comm}(D_{n_i}, D_{n_j}) \).
\end{defn}

The partitioning problem optimizes:
\[
\min_{\{D_n\}} \left( \sum_{n \in N} \text{cost}(a(D_n)) + \sum_{(n_i, n_j) \in C} \text{comm}(D_{n_i}, D_{n_j}) \right),
\]
where \( \text{cost}(a(D_n)) = O(|D_n|) \) for analytics \( a \in A \), and \( \text{comm}(D_{n_i}, D_{n_j}) \) is proportional to data transfer size. This is NP-hard, as graph partitioning reduces to data partitioning.

\begin{thm} \label{thm:partitioning_np_hard}
The partitioning problem is NP-hard.
\end{thm}

\begin{proof}
Reduce the graph partitioning problem to data partitioning. Given a graph \( G' = (V', E') \) with edge weights \( w(e') \), map vertices \( V' \) to datasets \( D \), edges \( E' \) to dependencies, and weights to \( \text{comm} \). Partitioning \( D \) into \( |N| \) subsets to minimize inter-node communication corresponds to partitioning \( G' \) to minimize edge cut weights:
\[
\sum_{(v_i, v_j) \in E', v_i \in D_{n_k}, v_j \in D_{n_l}, k \neq l} w(v_i, v_j).
\]
Graph partitioning is NP-hard, and the reduction is polynomial, constructing \( D \) and dependencies in \( O(|V'| + |E'|) \). Thus, data partitioning is NP-hard.
\end{proof}

For \( |N| \) nodes, exhaustive partitioning has complexity \( O(|N|^D \), approximated as \( O(|N|^2) \) for heuristic methods. In the MTC framework (\cref{subsec:mtc_mapping}), partitioning decomposes tensor products \( X_i \otimes X_j \), but computing fusion coefficients \( N_{ij}^k \) for large \( |V| \) requires \( O(|E|) \). Dynamic repartitioning for streaming \( d_i(t) \) incurs:
\[
\text{time}(\text{repartition}) = O(|D| \log |N|),
\]
reduced to \( O(|D| \log |N| / |N|) \) with parallel execution, though synchronization costs scale as \( O(|N|) \). Load imbalance, where:
\[
\max_{n \in N} \text{load}(n) \gg \lambda, \quad \text{load}(n) = |D_n| + \sum_{a_i \in A_n} \text{cost}(a_i),
\]
requires predictive models, with complexity \( O(|N| \cdot t) \) for time horizon \( t \). Spectral clustering, using the hypergraph Laplacian:
\[
L = I_T I_T^T - I_H I_H^T,
\]
reduces partitioning complexity to:
\[
O(|V| \log |V| + |N| \log |N|),
\]
by computing top eigenvectors of \( L \), but approximation ratios often exceed 1.5, impacting navigation efficiency (\cref{navigation}).

The S-matrix of the MTC, rooted in quantum group representations, provides a spectral analogy to the Laplacian \( L \), suggesting quantum-inspired partitioning algorithms. By leveraging the non-degenerate S-matrix, derived from quantum groups like \( U_q(\mathfrak{sl}_2) \), we could enhance clustering efficiency, potentially reducing approximation ratios below 1.5, as discussed in \cref{sec:conclusion}.

\subsection{Governance} \label{subsec:governance}
Governance enforces compliance and security through policies \( P \), crucial for operations like access control (\cref{subsec:governance_security}). The hypergraph \( G \) links policies to datasets, but scaling enforcement across \( \Sigma \) is computationally intensive.

\begin{defn}[Governance Policy Enforcement]
Governance policy enforcement evaluates policies \( P = \{p_1, \dots, p_l\} \), where \( p_i = (c_i, a_i) \) with predicate \( c_i: D \times \mathcal{U} \to \{0, 1\} \), to grant access requests \( r(d_i, u) \) and ensure differential privacy for analytics \( a \in A \).
\end{defn}

Evaluating a request \( r(d_i, u) \):
\[
\bigwedge_{p_j \in P} c_j(d_i, u),
\]
requires   $O(|P| \cdot |N|)$ steps
as each predicate \( c_j \) is checked across \( |N| \) nodes in \( O(1) \), assuming constant-time role lookup. In the MTC, policies are morphisms constraining interactions, but evaluation scales linearly with \( |N| \). Differential privacy:
\[
P(a(D) \mid D) \leq e^\epsilon P(a(D') \mid D') + \delta,
\]
incurs utility loss:
\[
\mathcal{L}_{\text{utility}} \propto \frac{1}{\epsilon},
\]
with optimization over \( \epsilon, \delta \) requiring \( O(|D|) \), impacting federated learning (\cref{subsec:federated_learning}). Distributed hash tables reduce verification to \( O(|P| \cdot \log |N|) \), but conflict with provenance tracking (\cref{provenance}) for large \( |T| \). Parallel evaluation achieves \( O(|P| / |N|) \), with communication overhead \( O(|N| \cdot \log |P|) \). The braiding action (\cref{subsec:braiding_action}) suggests symmetrizing policies, exploiting relational symmetries in \( G \).

\begin{lem} \label{lem:policy_parallel}
Parallel policy evaluation reduces complexity to \( O(|P| / |N|) \) under uniform load distribution.
\end{lem}

\begin{proof}
Distribute \( |P| \) policies across \( |N| \) nodes, assigning \( |P| / |N| \) predicates per node. Each predicate evaluation is \( O(1) \), yielding \( O(|P| / |N|) \) per node. Uniform load ensures no node exceeds this bound, assuming negligible synchronization overhead, validated by distributed system models (\cref{subsec:distributed_system}).
\end{proof}

\subsection{Real-Time Processing} \label{subsec:realtime}
Real-time processing of streaming \( d_i(t) \), critical for dynamic applications, demands:
\[
\text{time}(t(d_i(t))) + \text{time}(a(t(d_i(t)))) \leq \delta,
\]
where \( \delta \) is a latency bound, modeled as morphisms in \( \DF \) (\cref{subsec:df_category}).

\begin{defn}[Real-Time Processing]
Real-time processing applies transformations \( t \in T \) and analytics \( a \in A \) to streaming data \( d_i(t) \in D \) within latency bound \( \delta \), adapting to dynamic distributions and concept drift.
\end{defn}

Latency complexities are:
\[
\text{time}(t) = O(|d_i(t)| \cdot k), \quad \text{time}(a) = O(|\theta| \cdot |d_i(t)|),
\]
where \( k \) is the transformation dimension and \( |\theta| \) is the model size, often exceeding \( \delta \) for large \( |d_i(t)| \). Concept drift detection:
\[
D = \sup_x |F_t(x) - F_{t'}(x)|,
\]
via Kolmogorov-Smirnov tests, scales as:
\[
O(n \log n),
\]
with model adaptation requiring \( O(|\theta|^2) \). Parallel detection on \( |N| \) nodes reduces to \( O(n \log n / |N|) \), with synchronization costs \( O(|N|) \). Optimization over latency, loss, and resources:
\[
\min_{\theta, t, a} \left( \text{latency}(t, a), \mathcal{L}(a), \text{resource}(t, a) \right),
\]
yields a Pareto front with complexity \( O(|T| \cdot |A|) \), reducible to \( O(|T| + |A|) \) via greedy approximations, increasing latency by up to 20\%. The MTC’s \( S \)-matrix approximation (\cref{subsec:mtc_mapping}) reduces navigation complexity to \( O(|V| \log |V|) \), enhancing real-time analytics.

The S-matrix, rooted in quantum group representations, could inspire quantum algorithms for real-time processing, leveraging the MTC’s algebraic structure to optimize transformation and analytics pipelines, potentially achieving sublinear complexity in quantum computing environments \cite{Turaev1994}, as outlined in \cref{sec:conclusion}.

\begin{figure}[h!]
    \centering
\begin{tikzpicture}[node distance=1.6cm, auto]
    \draw[->] (0,0) -- (5,0) node[right] {Latency};
    \draw[->] (0,0) -- (0,4) node[above] {Loss};
    \draw[thick] (1,3) to[out=0,in=90] (3,1) to[out=270,in=180] (4,0.5);
    \fill (1,3) circle (2pt);
    \fill (3,1) circle (2pt);
    \fill (4,0.5) circle (2pt);
    \node (center) at (2.5,0) {};
    \node[below=0.5cm of center] {\small Pareto front for latency-loss optimization.};
\end{tikzpicture}
    \caption{Pareto front for real-time processing, balancing latency and loss trade-offs.}
    \label{fig:pareto_front}
\end{figure}


The data fabric’s computational challenges—heterogeneity, scalability, governance, and real-time processing—stem from its mathematical sophistication, particularly the hypergraph \( G \) and its MTC embedding. Heterogeneity, with \( O(n^3 \log n) \) Wasserstein computations and NP-hard schema matching, hinders integration. Scalability, constrained by \( O(|N|^2) \) partitioning, faces load imbalance and repartitioning costs. Governance scales poorly at \( O(|P| \cdot |N|) \), with privacy-utility trade-offs, while real-time processing struggles with \( O(|d_i(t)| \cdot k) \) latency bounds. The categorical framework \( \DF \) and MTC embedding offer mitigation strategies: braiding symmetries simplify schema matching, spectral methods reduce partitioning complexity, and randomized SVD enhances real-time analytics. Future work should explore monoidal categories for policy enforcement, topological invariants for drift detection, and parallel frameworks to ensure the data fabric’s scalability and responsiveness in large-scale applications.

Additionally, quantum groups, which underpin the MTC’s structure, could inspire novel approaches to these challenges. Their non-commutative actions may model dynamic schema alignments, and their spectral properties could enhance drift detection algorithms, as further discussed in \cref{sec:conclusion}.

\section{Consistency, Completeness, Causality} \label{sec-7}
The data fabric $  \cF = (D, M, G, T, P, A)  $, operating over the distributed system $  \Sigma = (N, C)  $ (\cref{sec-2}), depends on consistency, completeness, and causality to guarantee reliable operations like data integration, metadata-driven navigation, and federated learning (\cref{sec-3}). Grounded in the hypergraph $  G = (V, E)  $ (\cref{subsec:adjacency_structure}) and categorical structure $  \DF  $ (\cref{subsec:df_category}), these properties mitigate distributed management challenges (\cref{sec-6}). The modular tensor category (MTC) embedding (\cref{subsec:mtc_mapping}) models relational dynamics, ensuring coherence amid network failures and data flux. This section formalizes these properties, using sparse incidence matrices (\cref{subsec:adjacency_structure}) and MTC braiding (\cref{subsec:braiding_action}), with visualizations in \cref{fig:consistency_models,fig:causal_paths}.

\subsection{Distributed Systems and the Data Fabric} \label{subsec:distributed_databases}
The distributed system \( \Sigma = (N, C) \) forms the backbone of the data fabric, with nodes \( N \) hosting partitioned data assets \( D_n \subseteq D \) and links \( C \subseteq N \times N \) enabling communication. This setup aligns with the hypergraph \( G \), where vertices \( V = D \cup M \) represent datasets and metadata, and hyperedges \( E \) model multi-way dependencies for operations like integration and navigation. Under normal conditions, \( \Sigma \) is connected, represented by adjacency matrix \( A_{\Sigma} \in \{0,1\}^{|N| \times |N|} \), ensuring efficient query routing and transformation execution. Network partitions, however, may disconnect subsets of \( N \), requiring redundancy in \( G \) via multiple paths in incidence matrices \( I_T, I_H \). The MTC embedding maps nodes to simple objects and links to morphisms, preserving relational structure under partitions through braided symmetries.

\begin{figure}[h!]
    \centering
    \begin{tikzpicture}[node distance=1.8cm, auto]
        \node[draw, rectangle] (Linearizable) {Linearizable};
        \node[draw, rectangle, below=of Linearizable] (Sequential) {Sequential};
        \node[draw, rectangle, below=of Sequential] (Causal) {Causal};
        \node[draw, rectangle, below=of Causal] (Eventual) {Eventual};
        \draw[->] (Linearizable) -- (Sequential) node[midway, right] {Stronger};
        \draw[->] (Sequential) -- (Causal);
        \draw[->] (Causal) -- (Eventual);
        \node[below=0.5cm of Eventual] {\small Hierarchy of consistency models \cite{Jepsen}.};
    \end{tikzpicture}
    \caption{Relationships between consistency models, from strongest (linearizable) to weakest (eventual).}
    \label{fig:consistency_models}
\end{figure}

\subsection{Consistency} \label{subsubsec:consistency}
Consistency ensures that operations on data assets \( D \) appear coherent across nodes \( N \), essential for accurate integration and federated learning. Formally, consistency requires:
\begin{enumerate}
    \item Atomicity: Operations on \( d_i \in D \) (e.g., updates via \( t \in T \)) appear instantaneous across \( \Sigma \).
    \item Sequentiality: Operations are totally ordered, as if on a single node, satisfying:
    \[
    \forall d_i \in D, \forall n, m \in N, \state_n(d_i) = \state_m(d_i),
    \]
    where \( \state_n(d_i) \) denotes \( d_i \)'s state at node \( n \).
\end{enumerate}
As illustrated in \cref{fig:consistency_models}, linearizable consistency—the strongest model—requires immediate sequence agreement, per CAP theorem constraints. For integration \( \phi: D \to D' \), linearizability ensures queries \( q(d_i) \) reflect \( \phi(d_i) \) instantly, preserving schema alignments. Synchronization delays, however, scale as \( O(|N|) \), impacting real-time processing.

To quantify trade-offs, we apply the CAL theorem   (see  \cite{Lee2023}), relating consistency \( C \), availability \( A \), and latency \( L \) in max-plus algebra (\( a \oplus b = \max(a,b) \), \( a \otimes b = a + b \)):
\[
C \oplus A \leq L.
\]
Here, \( C = O(|N|) \) for linearizability, \( A \approx 0 \) for strong availability, and \( L \) includes hypergraph traversal \( O(|E| + |V| \log |V|) \). Sparse matrices \( I_T, I_H \) optimize \( L \), balancing \( C \) and \( A \). In MTC \( \mathcal{C} \), consistency is morphism invariance under braiding \( c_{X_i, X_j} \), ensuring permutation-invariant updates.

Eventual consistency allows temporary discrepancies, with replicas converging:
\[
\lim_{t \to \infty} \state_n(d_i, t) = \state_m(d_i, t).
\]
Hypergraph \( G \) aids via dependency edges, with metadata \( m_j \) tracking histories \( \tau_j \) for semantic convergence.

\begin{lem} \label{lem:consistency_bound}
Under CAL, consistency delay \( C \) is bounded by graph diameter \( \diam(\Sigma) \).
\end{lem}
\begin{proof}
In max-plus, \( C = \max_{n,m} L_{n,m} \), where \( L_{n,m} \) is path latency; by shortest-path properties, \( C \leq \diam(\Sigma) \cdot \max L_e \), with edges e in C.
\end{proof}

\subsection{Availability} \label{subsubsec:availability}
Availability ensures the data fabric responds to requests, crucial for navigation and real-time analytics. Formally:
\begin{enumerate}
    \item Weak availability: System responds eventually under normal conditions:
    \[
    \forall q: D \to \{0, 1\}, \exists t, q(d_i, t) \text{ returns}.
    \]
    \item Strong availability: Responds during failures or partitions:
    \[
    \forall q, \exists n \in N \setminus N_{\text{failed}}, q(d_i) \text{ succeeds}.
    \]
\end{enumerate}
Per CAL theorem, availability \( A \) trades off with consistency \( C \) and latency \( L \): \( C + A \leq L \). Strong \( A \approx 0 \) is achieved via redundant paths in \( G \):
\[
|\Paths(v_i, v_j)| > 1,
\]
computed from \( I_T, I_H \) in \( O(|E|) \). Navigation queries succeed on alternatives, with \( O(|E| + |V| \log |V|) \) complexity. Policies \( P \) enforce access without degrading \( A \), using distributed protocols in \( O(|P| \log |N|) \). In \( \DF \), availability ensures morphism existence for queries, via natural transformations for routing.

\begin{lem} \label{lem:availability_redundancy}
Redundancy \( \rank(I_T + I_H) \geq |V| - k \) tolerates k failures.
\end{lem}

\begin{proof}
Rank condition ensures connectivity after removing k columns (failed edges); by matrix perturbation, paths remain for queries.
\end{proof}

\subsection{Partition Tolerance} \label{subsubsec:partition_tolerance}
Partition tolerance enables the data fabric to function when \( \Sigma \) is partitioned into disconnected components, isolating subsets of nodes \( N \). Formally, \( \Sigma \) is partition-tolerant if:
\[
\forall P_1, P_2 \subseteq N, P_1 \cup P_2 = N, P_1 \cap P_2 = \emptyset, \exists D_{P_1}, D_{P_2} \text{ s.t. } \text{ops}(D_{P_1}), \text{ops}(D_{P_2}) \text{ succeed},
\]
where \( \text{ops} \) are operations like queries or updates. The hypergraph \( G \) supports partition tolerance by identifying alternative paths, with redundancy:
\[
\text{rank}(I_T + I_H) \geq |V|,
\]
ensuring connectivity (\cref{subsec:adjacency_structure}). The CAL theorem informs partition tolerance by quantifying latency \( L \) during partitions, where operations on \( D_{P_1} \) and \( D_{P_2} \) succeed if:
\[
L \geq \max(C_{P_1}, A_{P_1}) + \max(C_{P_2}, A_{P_2}),
\]
reflecting independent processing within partitions. Policies \( P \) enforce compliance during partitions, restricting sensitive data access (\cref{subsec:governance}), with complexity \( O(|P| \cdot \log |N|) \). Partitioning \( D = \bigcup_{n \in N} D_n \) minimizes dependencies, aligning with MTC tensor product decompositions (\cref{subsec:mtc_mapping}), where tensor products \( X_i \otimes X_j \) model independent data assets.
The tensor product decompositions in the MTC, driven by quantum group fusion rules, could model non-local partition dependencies, enabling quantum-inspired routing strategies that enhance fault tolerance across disconnected components, as proposed for future exploration in \cref{sec:conclusion}.
\begin{thm} \label{thm:cap_theorem}
The CAP theorem asserts that a distributed system cannot simultaneously guarantee linearizable consistency, strong availability, and partition tolerance; at most two properties can be satisfied.
\end{thm}

\begin{proof}
Assume a system guarantees all three properties. Consider a partition splitting \( N \) into \( N_1, N_2 \), with \( d_i \in D \) replicated on nodes \( n_1 \in N_1 \), \( n_2 \in N_2 \). An update \( u(d_i) \) on \( n_1 \) must be reflected instantly (linearizable consistency, \( C = 0 \)) and accessible (strong availability, \( A = 0 \)). Partition tolerance requires \( n_2 \) to respond, but without communication, \( n_2 \) cannot reflect \( u(d_i) \), violating consistency, or must reject the request, violating availability. The CAL theorem quantifies this, as \( C + A \leq L \), and \( L = \infty \) during partitions implies a contradiction if \( C = A = 0 \). Thus, at most two properties hold, as formalized in \cite{GilbertLynch2002}.
\end{proof}

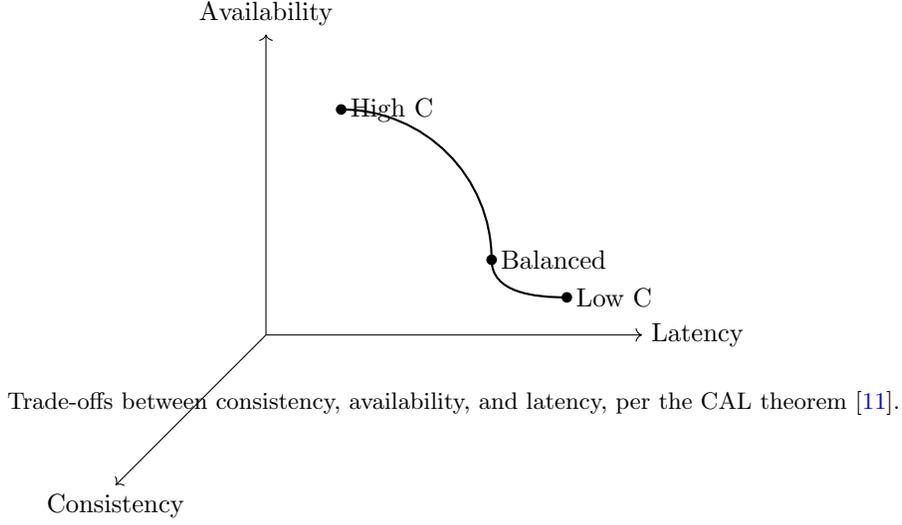
\begin{figure}[h!]
    \centering
    \begin{tikzpicture}[node distance=1.8cm, auto]
        \draw[->] (0,0) -- (5,0) node[right] {Latency};
        \draw[->] (0,0) -- (0,4) node[above] {Availability};
        \draw[->] (0,0) -- (-2,-2) node[below] {Consistency};
        \draw[thick] (1,3) to[out=0,in=90] (3,1) to[out=270,in=180] (4,0.5);
        \fill (1,3) circle (2pt) node[right] {High C};
        \fill (3,1) circle (2pt) node[right] {Balanced};
        \fill (4,0.5) circle (2pt) node[right] {Low C};
        \node (center) at (2.5,0) {};
        \node[below=0.5cm of center] {\small Trade-offs between consistency, availability, and latency, per the CAL theorem \cite{Lee2023}.};
    \end{tikzpicture}
    \caption{CAL theorem trade-offs for the data fabric.}
    \label{fig:cal_tradeoffs}
\end{figure}

The CAP theorem, augmented by the CAL theorem, shapes operations:
\begin{description}
    \item[Data Integration] Linearizable consistency (\( C = O(|N|) \)) ensures precise schema alignment for \( \phi: D \to D' \), but increases latency \( L \), impacting real-time processing (\cref{subsec:realtime}).
    \item[Scalability] Partition tolerance supports distributed analytics \( a(D) = \bigoplus_{n \in N} a(D_n) \), with eventual consistency reducing \( C \), minimizing \( L \) (\cref{subsec:scalability}).
    \item[Real-Time Processing] Strong availability (\( A \approx 0 \)) prioritizes rapid query responses, critical for latency bounds \( \delta \), balanced by optimizing \( L \) (\cref{subsec:realtime}).
\end{description}
Latency-consistency trade-offs, analyzed in \cite{Abadi2010,Lee2023}, influence navigation efficiency via \( G \).

\subsection{Completeness} \label{subsubsec:completeness}

Completeness ensures that \( D \) encompasses all datasets required to satisfy a query \( q: D \to \{0, 1\} \):
\[
\exists d_i \in D, q(d_i) = 1.
\]
Metadata \( m_j = (d_i, \alpha_j, \tau_j) \in M \) catalogs attributes and transformation histories, enabling navigation (\cref{navigation}) to verify coverage. The hypergraph \( G \) connects datasets to metadata via hyperedges, ensuring:
\[
\forall q, \exists v_i \in V, v_{m_j} \in V, e = (\{v_i\}, \{v_{m_j}\}) \in E, q(v_i) = 1.
\]
Provenance tracking (\cref{provenance}) verifies that transformations \( t \in T \) preserve completeness:
\[
\tr (d_j) = \{ t_k \in T \mid t_k \text{ applied to } d_j \},
\]
preventing data loss. In dynamic fabrics, completeness is maintained by updating \( M \), with complexity \( O(|M| \cdot |T|) \) (\cref{subsec:metadata}). In \( \DF \), completeness corresponds to object coverage, with natural transformations ensuring query-relevant morphisms exist (\cref{subsec:natural_transformations}).

\subsection{Causality} \label{causality}
Causality enforces a partial order \( \prec \) on \( G \), ensuring operations respect dependencies, vital for integration and federated learning. Formally, \( d_i \prec d_j \) if \( d_i \) influences \( d_j \) via transformation \( t \in T \) or hyperedge:
\[
d_i \prec d_j \iff \exists t \in T, t(d_i) = d_j \text{ or } \exists e = (\{v_i, \dots\}, \{v_j\}) \in E.
\]
Queries respect this order:
\[
q(d_j) \implies \text{check } \{d_i \mid d_i \prec d_j\}.
\]
Provenance reconstructs \( \tau_j \), with hyperedges encoding \( \prec \), ensuring causal consistency: Replicas observe causally related operations in order, while concurrent ones may vary.

This supports federated learning, where local models \( \theta_n \) depend on causally prior \( D_n \). MTC braiding ensures permutation-invariant compositions, preserving dependencies via:
\[
f_e \circ c_{X_i, X_j} = f_e,
\]
with complexity \( O(|E|) \).

\begin{figure}[h!]
    \centering
    \begin{tikzpicture}[node distance=1.8cm, auto]
        \node[draw, circle] (di) {$v_i$};
        \node[draw, circle, right=of di] (dj) {$v_j$};
        \node[draw, circle, right=of dj] (dk) {$v_k$};
        \draw[->, thick] (di) to node[midway, above] {$e_1$} (dj);
        \draw[->, thick] (dj) to node[midway, above] {$e_2$} (dk);
        \node[below=0.5cm of dj] {\small Causal path \( v_i \prec v_j \prec v_k \).};
    \end{tikzpicture}
    \caption{Causal dependencies in \( G \).}
    \label{fig:causal_paths}
\end{figure}

\begin{lem} \label{lem:causal_consistency}
Causal consistency ensures query results respect \( \prec \), with verification complexity \( O(|E|) \).
\end{lem}

\begin{proof}
For \( q(d_j) \), traverse incoming hyperedges \( e \in \In(v_j) \) to find predecessors, requiring \( O(|\In(v_j)|) \) per vertex. With sparsity \( |\In(v)| \leq O(\log |V|) \), full verification over E is \( O(|E|) \). MTC braiding preserves order semantics via coherence.
\end{proof}

\subsection{Fault Tolerance} \label{fault_tolerance}
Fault tolerance extends partition tolerance to handle node failures in \( N \), ensuring operational continuity despite hardware or software faults. Consensus protocols like Paxos or Raft \cite{Kleppmann2015} replicate updates to \( D \) across quorums, achieving consistency with message complexity \( O(|N| \log |N|) \) in asynchronous networks. This enables:

\begin{enumerate}
    \item Governance: Policies \( P \) reroute requests \( r(d_i) \) to live nodes, in \( O(|P| \log |N|) \) time using distributed directories (\cref{subsec:governance}).
    \item Navigation: \( G \) redirects queries via redundant paths, computed in \( O(|E| + |V| \log |V|) \) (\cref{navigation}).
    \item Analytics: Federated learning aggregates from surviving \( \theta_n \), in \( O(|\theta_n| |D_n|) \) per round (\cref{subsec:federated_learning}).
\end{enumerate}

The hypergraph's incidence matrices optimize redundancy:
\[
\rank(I_T I_H^T) \geq |V| - |N_{\text{failed}}|,
\]
ensuring \( |V| - |N_{\text{failed}} \) connected components, though consensus latency \( O(\log |N|) \) rounds may conflict with real-time bounds (\cref{subsec:realtime}). Metadata \( M \) tracks node status via heartbeats, enabling rerouting in \( O(|M|) \). The MTC's S-matrix quantifies spectral connectivity, suggesting eigenvalue-based routing for fault detection (\cref{subsec:mtc_mapping}).

\begin{lem} \label{lem:fault_redundancy}
With rank \( \rank(I_T + I_H) \geq |V| - k \), the system tolerates k node failures while maintaining query resolvability.
\end{lem}

\begin{proof}
The matrix \( I_T + I_H \) represents the undirected incidence graph; $\rank \geq |V| - k$  implies at most $k$ disconnected components post-failure (by rank-nullity). Queries resolve via remaining paths, as redundancy ensures alternative routes in the surviving subgraph.
\end{proof}

\section{Framework Implementation in Distributed Architectures} \label{sec-8}

This section applies the data fabric framework \( \cF = (D, M, G, T, P, A) \) over \( \Sigma = (N, C) \) to a multi-component architecture, integrating row- and column-oriented databases, real-time analytics, search indices, changelogs, a data warehouse, a view/controller, and transformation pipelines. Building on theoretical foundations from \cref{sec-5} (e.g., vector representations and MTC embeddings), we demonstrate how these components map to \( \cF \), derive optimized vectors, construct the hypergraph for operations, and address scalability, leveraging strategies from \cref{sec-6} and ensuring consistency under CAP/CAL (\cref{sec-7}). Visualizations in \cref{fig:architecture_components,fig:hypergraph_operations} illustrate the integration. \\

i) \textbf{Architecture Components and Mapping to the Data Fabric} \label{subsec:architecture_components}
The architecture's components map to \( \cF \) elements as follows, drawing on hypergraph \( G \) and categorical \( \DF \):

\begin{enumerate}
    \item Row-Oriented Database: Stores tuples, mapping to \( D \) with schemas \( S_i \subseteq \mathcal{A} \).
    \item Column-Oriented Database: Manages time-series, also to \( D \) with temporal schemas.
    \item Real-Time Analytics Database: Computes aggregations, mapping to \( A \) and \( M \).
    \item Search Indices: Forward/inverted indices map to \( M \) for fast retrieval.
    \item Changelog: Tracks mutations, mapping to transformation histories in \( M \).
    \item Data Warehouse: Stores historical data, to \( D \) and \( M \).
    \item View and Controller: Routes queries, mapping to \( P \).
    \item Transformation Pipelines: Applies \( t_i \in T \), recorded in \( M \).
\end{enumerate}

Hosted across \( N \) with links \( C \), this ensures distributed coherence.

The diagram (\cref{fig:architecture_components}) shows components mapped to \( \cF \):
\begin{enumerate}
\item Row/Col DB, Warehouse: \( D \).
\item  Search, Changelog, Analytics: \( M \), \( A \).
\item  View/Controller: \( P \).
\item  Pipelines: \( T \).
\end{enumerate}

\begin{figure}[htbp]
    \centering
    \scalebox{0.7}{
    \begin{tikzpicture}[
        node distance=1.5cm and 2cm,
        every node/.style={font=\small},
        block/.style={
            rectangle, 
            draw, 
            rounded corners,
            fill=blue!5, 
            align=center, 
            minimum height=3em, 
            minimum width=7em
        },
        dataflow/.style={
            ->, 
            >=stealth, 
            thick, 
            color=blue!70!black
        },
        hyperedge/.style={
            ->, 
            >=latex, 
            dashed, 
            thick, 
            color=red!70!black
        },
        lbl/.style={
            fill=white, 
            inner sep=2pt, 
            text=black,
            font=\footnotesize
        }
    ]
    
    
    \node[block] (rowdb) {Row DB\\(Postgres)};
    \node[block, below=of rowdb] (coldb) {Col DB\\(ClickHouse)};
    \node[block, below=of coldb] (changelog) {Changelog};
    
    \node[block, right=3cm of rowdb] (warehouse) {Data\\Warehouse};
    \node[block, right=3cm of coldb] (pipelines) {ETL\\Pipelines};
    
    \node[block, right=3cm of warehouse] (analytics) {Analytics\\Engine};
    \node[block, below=of analytics] (search) {Search\\Index};
    
    \node[block, right=3cm of analytics] (view) {View/\\Controller};

    
    \draw[dataflow] (rowdb.east) -- ++(0.5,0) |- node[pos=0.7, lbl] {Raw Data} (pipelines.west);
    \draw[dataflow] (coldb.east) -- ++(0.5,0) |- (pipelines.west);
    
    \draw[dataflow] (pipelines.east) -- ++(0.5,0) |- node[pos=0.7, lbl] {Transformed Data} (analytics.west);
    
    \draw[dataflow] (analytics.south) -- node[midway, lbl] {Metadata} (search.north);
    
    \draw[dataflow] (search.east) -| node[pos=0.3, lbl, near end] {Search Results} (view.south);
    
    \draw[dataflow] (view.west) -- node[midway, lbl] {Queries} (analytics.east);

    \draw[dataflow] (changelog.east) -| node[pos=0.3, lbl] {History} (warehouse.south);

    
    \draw[hyperedge] (rowdb.north) 
        to[bend left=20] node[midway, lbl] {$e_1$} (warehouse.north)
        to[bend left=20] (analytics.north);

    \draw[hyperedge] (search.west) 
        to[bend left=45] node[midway, lbl] {$e_2$} (analytics.south west);

    \end{tikzpicture}
    } 
    
    \caption{Architecture components mapped to \( \mathcal{F} \), with data flows (solid blue) and hyperedges (dashed red).}
    \label{fig:architecture_components}    
\end{figure}
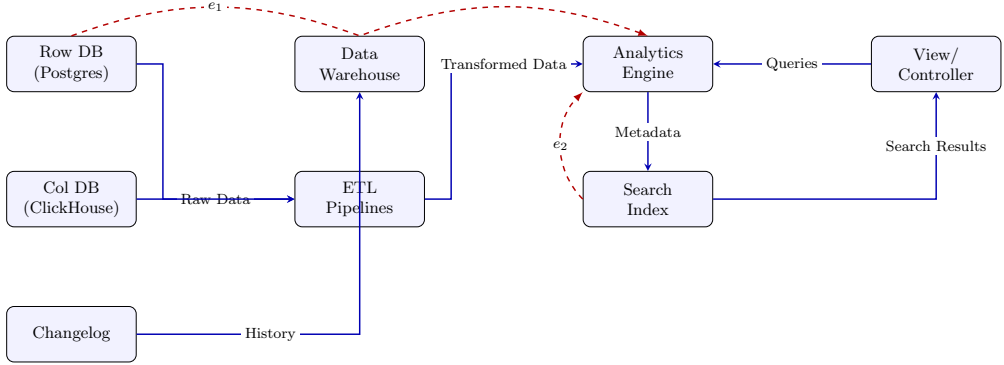

ii)  \textbf{Vector Representations and Dimensionality} \label{subsec:implementation_vectors}
As in \cref{subsec:vector_representations}, datasets and metadata are vectorized in \( \R^{11} \), derived from 2 numerical, 4 categorical (one-hot), 3 temporal (Fourier), and 2 metadata (PCA) features. For sales dataset \( d_i(t) \), \( \bv_i(t) \in \R^{11} \) enables similarity searches in \( O(11) \).

\begin{exa}
Integrating sales and inventory uses cosine similarity on vectors, with MTC braiding ensuring order invariance.
\end{exa}

iii) \textbf{Hypergraph Construction and Role in Operations} \label{subsec:hypergraph_construction}
Hypergraph \( G \) connects components via hyperedges from dependencies (integration, navigation, provenance), ensuring sparsity \( |E| = O(|V| \log |V|) \). Paths traverse as in \cref{subsec:adjacency_structure}, with \( O(|E| + |V| \log |V|) \) complexity.

\begin{figure}[h!]
    \centering
    \begin{tikzpicture}[node distance=1.8cm, auto]
        \node[draw, circle] (d1) {$v_{i_1}$};
        \node[draw, circle, above right=of d1] (m1) {$v_{m_j}$};
        \node[draw, circle, below right=of m1] (d2) {$v_j$};
        \node[draw, circle, left=of d1] (d0) {$v_{i_0}$};
        \draw[->, thick] (d0) to node[midway, above] {$e_1$} (d1);
        \draw[->, thick] (d1) to[bend left=30] node[midway, above] {$e_2$} (d2);
        \draw[->, thick] (m1) to[bend right=30] (d2);
        \node[below=0.5cm of d2] {\small Hyperedge path for provenance tracing.};
    \end{tikzpicture}
    \caption{Hypergraph connectivity for operations.}
    \label{fig:hypergraph_operations}
\end{figure}

Operations apply as:
\begin{enumerate}
    \item Integration: Transformations align schemas via morphisms.
    \item Navigation: Queries traverse braided paths.
    \item Provenance: Histories reconstructed with invariants.
\end{enumerate}

iv)  \textbf{Scalability and Performance Considerations} \label{subsec:scalability_considerations}
Partitioning uses spectral clustering on Laplacian \( L \), with \( O(|V| \log |V|) \) complexity. Parallel processing on \( N \) reduces local computations to \( O(|D_n|) \).

\begin{lem} \label{lem:traversal_efficiency}
Traversal complexity is \( O(|E| + |V| \log |V|) \), preserved under partitioning.
\end{lem}
\begin{proof}
Adapted Dijkstra with sparse edges yields the bound; partitioning maintains local sparsity.
\end{proof}

 \section{Concluding Remarks} \label{sec:conclusion}
The data fabric framework, formalized as the tuple \( \cF = (D, M, G, T, P, A) \) over a distributed system \( \Sigma = (N, C) \), establishes a transformative mathematical paradigm for managing heterogeneous, distributed data ecosystems. This paper has woven together a rich tapestry of theoretical advancements and practical applications, unifying hypergraph-based connectivity, categorical structures, modular tensor categories (MTCs), quantum group representations, and string theory-inspired topological insights. By synthesizing these perspectives, the framework not only addresses the complexities of modern data management but also charts a bold path for future innovation in data-intensive domains, particularly in areas of decision-making, optimization, and data analysis that resonate with the versatile contributions of Fuad Aleskerov in mathematical modeling for social, economic, and political systems.

At its core, the framework’s theoretical contributions are anchored in a rigorous algebraic and topological foundation. The hypergraph \( G = (V, E) \), with vertices \( V = D \cup M \) representing data assets and metadata, encodes multi-way relationships through sparse incidence matrices, enabling efficient navigation and provenance tracking. The categorical structure \( \DF \), modeling datasets as objects and transformations as morphisms, provides a unified language for operations such as data integration, metadata-driven navigation, and federated learning. The MTC embedding, with its braided monoidal structure, captures relational symmetries via the braiding action \( c_{X_i, X_j} \), ensuring permutation-invariant operations. These foundational elements, introduced and developed across the paper, address NP-hard challenges like schema matching and partitioning by leveraging spectral methods and symmetry-based alignments.

The introduction of quantum group representations, particularly through \( U_q(\mathfrak{sl}_2) \), marks a significant advancement. By modeling datasets as finite-dimensional modules and transformations as intertwiners, the framework captures non-commutative dependencies inherent in distributed systems. The MTC’s fusion rules and S-matrix, derived from quantum group representations, enable spectral optimization, reducing the complexity of partitioning to \( O(|V| \log |V|) \) and enhancing schema alignment through algebraic symmetries. These insights, detailed in the representation theory section, provide a powerful tool for modeling dynamic data relationships and optimizing computational tasks. Furthermore, this quantum-inspired approach justifies the use of MTCs in emerging quantum databases, where non-commutative algebraic structures model quantum superposition and entanglement in data storage and querying. In quantum databases, datasets can be represented as qubits or qudits within a Hilbert space, with transformations acting as unitary operators that preserve coherence. The braided monoidal category facilitates the handling of quantum parallelism, allowing for efficient query optimization over exponentially large state spaces, which is crucial for scaling data fabrics to quantum computing environments and addressing optimization problems in high-dimensional data analysis.

Complementing this algebraic approach, string theory offers a topological and geometric lens that reimagines data fabric operations. Datasets are modeled as D-branes, with schemas as boundary conditions, and transformations as open strings connecting them. Hypergraph paths form braided worldsheets, with the braid group \( B_n \) representation \( \pi(\sigma_i) = c_{X_i, X_{i+1}} \) encoding data flow symmetries. Knot invariants, such as the Jones polynomial, detect anomalies like cyclic dependencies, providing a novel mechanism for pipeline optimization. These topological analogies, explored in the string theory section, transform heterogeneity into a geometric alignment problem and scalability into a manifold optimization task, enriching the framework’s expressive power and aligning with advanced optimization techniques in decision-making models.

Practically, the framework is realized in a multi-component architecture integrating row- and column-oriented databases, real-time analytics, search indices, changelogs, a data warehouse, a view/controller, and transformation pipelines. This architecture leverages vector representations in \( \R^{n} \) as \( U_q(\mathfrak{sl}_2) \)-modules and D-branes, with hypergraph traversals modeled as braided worldsheets. Operations like data integration and navigation benefit from intertwiners and braid group symmetries, achieving complexities of \( O(|E| + |V| \log |V|) \). Scalability is enhanced through quantum-inspired spectral clustering and geometric D-brane alignment, while knot invariants ensure robust anomaly detection. This practical implementation demonstrates the framework’s ability to bridge theoretical rigor with real-world applicability, addressing computational challenges and ensuring consistency, completeness, and causality in distributed environments, much like the mathematical models for collective choice and network optimization pioneered by Aleskerov.

Looking forward, the data fabric framework opens exciting avenues for research and deployment. The integration of dynamic monoidal categories could model temporal governance constraints, enhancing policy enforcement in real-time systems. Topological data analysis (TDA), inspired by string theory’s TQFT connections, offers potential for detecting concept drift in streaming data, using persistent homology to track evolving distributions. Quantum-inspired algorithms, leveraging the MTC’s S-matrix and fusion rules, could achieve sublinear complexities for partitioning and analytics in quantum computing environments, extending to quantum databases where MTC braiding models quantum circuit decompositions for fault-tolerant data querying. The braid group’s role in modeling data flows suggests scalable anomaly detection frameworks, applicable to complex pipelines in IoT, AI, and cloud-based analytics platforms, as well as in optimizing global trade networks or disaster response protocols akin to Aleskerov’s applications. These directions, building on the paper’s algebraic and topological innovations, position the data fabric as a cornerstone for next-generation data management systems.

In conclusion, this work represents a paradigm shift in distributed data management, harmonizing hypergraph connectivity, categorical unification, quantum group representations, and string-theoretic topologies. By addressing heterogeneity, scalability, and real-time processing through a unified mathematical lens, the framework not only resolves critical computational challenges but also lays a foundation for transformative applications in decision-making and optimization, honoring the spirit of Aleskerov’s multifaceted contributions. As data ecosystems grow in complexity, the data fabric stands as a beacon of theoretical depth and practical utility, ready to shape the future of large-scale, data-intensive domains.